\documentclass{article} 
\usepackage{iclr2027_conference,times}

\usepackage{amsmath,amsfonts,bm}

\def\eqref#1{equation~\ref{#1}}

\def\1{\bm{1}}

\DeclareMathAlphabet{\mathsfit}{\encodingdefault}{\sfdefault}{m}{sl}
\SetMathAlphabet{\mathsfit}{bold}{\encodingdefault}{\sfdefault}{bx}{n}

\def\gM{{\mathcal{M}}}

\def\gT{{\mathcal{T}}}

\def\gV{{\mathcal{V}}}

\usepackage[utf8]{inputenc}
\usepackage[T1]{fontenc}
\usepackage{microtype}
\usepackage{amsmath,amssymb,amsthm,mathtools}
\usepackage{graphicx}
\usepackage{booktabs}
\usepackage{longtable}
\usepackage{array}
\usepackage{multirow}
\usepackage{tabularx}
\usepackage{xcolor}
\usepackage{tikz}
\usetikzlibrary{arrows.meta,positioning,fit,shapes.geometric}
\usepackage{enumitem}
\usepackage{algorithm}
\usepackage{algpseudocode}
\usepackage{hyperref}
\usepackage[capitalize,noabbrev]{cleveref}
\usepackage{url}
\usepackage{wrapfig}
\usepackage{subcaption}
\usepackage{dsfont}
\usepackage{etoc}
\usepackage{enumitem}
\usepackage{diagbox}
\usepackage{xspace}
\usepackage{adjustbox}
\usepackage{caption}

\usepackage{algorithmicx}
\usepackage{algpseudocodex}

\usepackage{setspace}

\usepackage{float}
\usepackage[most]{tcolorbox}
\definecolor{pbFrame}{HTML}{1A1AD2}   
\definecolor{pbTitle}{HTML}{0B0BA8}   
\definecolor{pbBack} {HTML}{F6F6FF}   
\definecolor{citeblue}{RGB}{0, 32, 128}
\usepackage{colortbl} 

\definecolor{guardblue}{HTML}{285C8E}
\definecolor{guardteal}{HTML}{2B7A78}
\definecolor{guardred}{HTML}{A23B3B}
\definecolor{guardgray}{HTML}{F1F3F5}
\definecolor{guardgold}{HTML}{B07A16}

\newcommand{\method}{\texttt{\textbf{PACE}}\xspace}
\newcommand{\systemhook}[1]{\texttt{\small #1}}
\newcommand{\Ctx}{\mathsf{Ctx}}
\newcommand{\Traj}{\mathsf{Traj}}
\newcommand{\Egress}{\mathsf{egress}}
\newcommand{\Cert}{\mathsf{cert}}

\newcommand{\terminalbad}{T_{\mathrm{bad}}}
\newcommand{\basefamily}{\mathsf{Beh}}
\newcommand{\tmpl}{\omega}
\newcommand{\gate}{\Gamma}
\newcommand{\cut}{\mathcal{C}}
\newcommand{\src}{\mathsf{src}}
\newcommand{\snk}{\mathsf{snk}}
\newcommand{\bind}{\beta}
\newcommand{\Sev}{\mathsf{Sev}}
\newcommand{\Impl}{\mathsf{Impl}}
\newcommand{\tightparagraph}[1]{\vspace{0.2em}\noindent\textbf{#1}}
\providecommand{\manifest}{\gM}
\definecolor{ceruleanblue}{rgb}{0.16, 0.32, 0.75}
\definecolor{kleinblue}{rgb}{0,0.18,0.65}
\hypersetup{
 colorlinks=true,
 citecolor=ceruleanblue,
 linkcolor=red,
 urlcolor=black}

\makeatletter
\DeclareRobustCommand{\Rmnum}[1]{\expandafter\@slowromancap\romannumeral #1@}  
\makeatother

\newtheorem{theorem}{Theorem}
\newtheorem{proposition}{Proposition}
\newtheorem{corollary}{Corollary}
\newtheorem{lemma}{Lemma}
\theoremstyle{definition}
\newtheorem{definition}{Definition}

\theoremstyle{remark}
\newtheorem{remark}{Remark}
\crefname{assumption}{assumption}{assumptions}
\Crefname{assumption}{Assumption}{Assumptions}
\crefname{definition}{definition}{definitions}
\Crefname{definition}{Definition}{Definitions}
\crefname{theorem}{theorem}{theorems}
\Crefname{theorem}{Theorem}{Theorems}
\crefname{proposition}{proposition}{propositions}
\Crefname{proposition}{Proposition}{Propositions}
\crefname{corollary}{corollary}{corollaries}
\Crefname{corollary}{Corollary}{Corollaries}

\title{\method{}: Provenance-Aware Capability Enforcement for Tool-Using LLM Agents}

\author{Fengpeng Li$^{1}$ \hspace{0.25cm} Qizhou Wang$^{2}$ \hspace{0.25cm}Yuke Hu$^{1}$ \hspace{0.25cm} Kemou Li$^3$ \hspace{0.25cm} \textbf{Jun Liu}$^4$ \hspace{0.25cm} \textbf{Haiwei Wu}$^5$ \hspace{0.25cm} \\ \textbf{Jiantao Zhou}$^3$\hspace{0.25cm} \textbf{Di Wang}$^1$\thanks{Corresponding author: Di Wang (di.wang@kaust.edu.sa)}\vspace{1.5mm}\\
$^1$PRADA Lab, King Abdullah University of Science and Technology \\
$^2$Imperfect Information Learning Team, RIKEN Center for Advanced Intelligence Project\\
$^3$State Key Laboratory of Internet of Things for Smart City, University of Macau \\
$^4$National Institute of Informatics\\
$^5$School of Computer Science and Engineering, University of Electronic Science and Technology of China
}

\iclrfinalcopy
\begin{document}

\maketitle
\lhead{Preprint} 
\begin{abstract}
Tool-using large language model (LLM) agents turn generated text into real side
effects, so poisoned tool metadata, retrieved pages, memory, and reusable skills
can steer the next call. Vetting an artifact before admission does not settle
this. A safe variant and a leaking variant can produce the same admission
evidence, and a sound gate then cannot relax that site for either. We make that
condition precise, which leaves the last boundary a deployment can still act on.
We present Provenance-Aware Capability Enforcement (\method{}), which mediates every tool
call immediately before it executes. Path confinement proposes an executable cut of represented
influence paths, while capability and effect verification checks schema-defined
effects against authority compiled from the authenticated request. We distinguish
the certified execution contract from the evaluated configuration, which can
restore an authorized call after a proposed block or apply a declared repair.
Confinement requires the final action to preserve the certified cut. On
eight executable agent-security benchmarks with three target-model families,
the evaluated configuration gives strictly lowest attack success in
\textbf{62 of 79 eligible attack columns} and ties in 14; full-benchmark native
utility loses at most three points relative to the undefended agent. A complete ablation
over 1167 paired cases attributes most security gains to effect verification
and refusal control to boundary adaptation. A reduced-scale adaptive search
succeeds on 0/30 out-of-authority targets against the defense. Source code is in the supplemental material.
\end{abstract}

\etocdepthtag.toc{mtchapter}

\section{Introduction}
\label{sec:introduction}

Tool-using LLM agents plan, call tools, and act inside one loop
\citep{react,toolformer,toolllm,gorilla}. The same loop reads retrieved pages,
tool returns, files, and third-party skill descriptions, so instructions and data
share the context that picks the next action \citep{zverev2025sep}. Prompt
injection exploits this across tool suites, web tasks, and multi-step workflows
\citep{perez2022,greshake2023,injecagent,agentdojo,asb,wasp,liu2024formalizing,yi2025bipia,toolemu},
with payloads arriving through MCP schemas \citep{mcptox,mcpsecbench,hou2025mcp},
reusable skills \citep{skillattack}, or poisoned memory and knowledge bases
\citep{agentpoison,poisonedrag,pasb}. Once the agent can send mail, write files,
or spend money, that text can cause a real effect. We target
deployments where it also reaches the planning context, as MCP tool descriptions
and marketplace skills do today. \Cref{app:related-work} surveys related work.

Defenses fall into four families: flow control anchored in a trusted plan or a
reconstructed trace \citep{camel,fides,rtbas,agentarmor}, declarative privilege
and task alignment \citep{progent,taskshield}, detection and re-execution
\citep{melon,datasentinel,llamafirewall,nemoguardrails,llamaguard}, and
instruction/data separation
\citep{wallace2024hierarchy,struq,secalign,hines2024spotlighting}. Their trusted inputs differ: plans, request-derived policies,
instrumented provenance, or learned distinctions between instructions and data
(\Cref{app:related-work}). \emph{Trust inversion} arises when the authority for admission is
itself attacker-controlled. An admission gate may review a skill by reading its specification
and the agent later follows that same specification, so the object under review
and the authority used to review it are one attacker-written string. Published
attacks on tool metadata exhibit this confused deputy \citep{mcptox};
adversarial prompting of reusable skills exposes a related surface
\citep{skillattack}. Dynamic labels track instrumented flows but do not
certify behavior over attacker-controlled contexts
\citep{taintdroid,libdft}. Underneath sits a mismatch of evidence:
confidentiality compares two runs and needs the semantics of both
\citep{clarkson2010,terauchi2005,barthe2004}, while a reviewer holds evidence
about one \citep{necula1997,sekar2003}. \Cref{sec:motivation} makes this precise
over the \emph{behavior family}, the trajectory distributions an artifact induces
across attacker-supplied contexts, and gives the condition under which no
admission-time evidence licenses relaxing a site.

We therefore decide at the last tool boundary before an external effect. Four
things meet there: the concrete proposal, the authenticated request, the observed
provenance, and a hook that can still change the outcome. \method{} pairs path
confinement with capability and effect verification. The first proposes actions
that sever represented influence paths; the second checks request authority over
concrete effects. Their composition depends on the final action: restoring a
blocked call can discard separation, and rewriting an authorized call can change
its effects. The certified contract preserves the cut and rechecks the final
call; the evaluated configuration permits restoration and declared repairs.
Recent provenance and execution monitors already check argument provenance
and fine-grained authority \citep{pact2026,authgraph2026,cxi2026}. Our analysis identifies the
graph-separation condition an intervention must retain alongside effect authority
(\Cref{sec:guarantee}), within the flows and effects the schemas represent.

We make the following contributions.
\begin{itemize}[leftmargin=*,itemsep=0pt,topsep=2pt,parsep=0pt]
  \item The security-relevant object is an artifact's behavior family over
attacker-controlled contexts. We show when admission-time evidence cannot
license relaxing an egress site.

  \item We propose \method{}, pairing path confinement with capability and effect verification
at the last enforceable tool boundary, and distinguish its certified contract
from the restoration and repair behavior of the evaluated configuration.

  \item We prove
  represented-path separation, authorized-effect soundness, and a conditional
  grounded-execution theorem, and specify which final interventions retain its
path-separation premise. These guarantees concern represented chains.

  \item On eight executable benchmarks with three target-model families,
  the evaluated configuration achieves strictly lowest ASR in 62 of 79 eligible
columns and ties in 14, with a paired component study and a reduced-scale
adaptive evaluation.
\end{itemize}

\section{Preliminaries}
We define the agent execution model and the security property used throughout the analysis.
\subsection{Agent Execution and Threat Model}
\label{sec:execution-model}
A frozen policy $M$ alternates free text with tool calls. An \emph{artifact} $K$
is a skill, a tool description, or an MCP schema; paired with a context
$c\in\Ctx$ carrying the request, injected text, tool returns, state, and possibly
a secret $s$, it induces a trajectory measure $\mu_{K,c}$ (\Cref{app:measure}).
The security-relevant object is therefore not $K$ but the \emph{behavior family}
$\basefamily(K)=\{\mu_{K,c}:c\in\Ctx\}$, which no audit of $K$ enumerates. The
attacker authors $K$ and any untrusted span of $c$, adapts to each defense
outcome, and succeeds on an unrequested irreversible effect or on a disclosure of
$s$. The defender acts only at the tool boundary, reading proposals and returns
and weakening or blocking a call, and it cannot retrain $M$, instrument a tool's
interior, or enumerate $\Ctx$.

\subsection{The Security Property and Why It Is Hard}
\label{sec:security-property}
Probabilistic noninterference \citep{goguen1982,sabelfeld2003} is the yardstick
we measure against, not a property we claim. \Cref{sec:motivation} uses it to
state what admission-time evidence cannot license, and \Cref{sec:guarantee} says
how far short of it \method{} stops. Let $\Egress(\tau)$ list the egress events of
a trajectory in order over a finite site set $\Omega$, with the answer $r$ as one
site, and let $\mu^{\gate}_{K,c}$ be the measure induced when $M$ runs behind a
mediator $\gate$. Two contexts differing only in $s$ should then give egress no
observer can tell apart beyond $\delta$:
\begin{equation}
  \Phi_{\mathrm{NI}}(K,\gate;\delta):\quad
  \sup_{\bar c\in\overline{\Ctx}}\ \sup_{s_0,s_1}\
  d\bigl(\Egress_{\#}\mu^{\gate}_{K,\bar c[s\leftarrow s_0]},\;
         \Egress_{\#}\mu^{\gate}_{K,\bar c[s\leftarrow s_1]}\bigr)
  \;\leq\;\delta.
  \label{eq:phi-ni}
\end{equation}
The projection keeps raw bytes, so base64, hashes, paraphrase, and fragmentation
sit inside the comparison and need no separate leak predicate. Finite observations alone do not discharge \Cref{eq:phi-ni} over an
unrestricted behavior family. It compares induced measures rather than isolated
traces, so a finite sample is not a conclusive witness of a distributional
violation. Hyperproperty reasoning \citep{clarkson2010} and self-composition
\citep{terauchi2005,barthe2004} require an adequate semantic model of the paired
executions; an artifact's code slice need not model the agent's choice of inputs.

\section{Limits of Admission-Time Certification}
\label{sec:motivation}
Consider a mail skill whose code always calls
$\texttt{send\_email}(\mathit{body})$, with a natural-language specification that
tells the model how to build $\mathit{body}$. A benign specification uses a fixed
template, a malicious one asks the model to copy a credential out of context, and
the code and interface are identical either way. That isolates the question:
\emph{when does evidence collected before execution justify removing mediation
from an egress site?}

Fix an egress site $\tmpl\in\Omega$ and let $O_{\tmpl}(\tau)$ be the subsequence
of $\Egress(\tau)$ tagged $\tmpl$, or $\bot$ if the site is never reached. It is
measurable (\Cref{lem:site-necessary}), so restricting \Cref{eq:phi-ni} to it
gives a local criterion:
\begin{equation}
  \Phi_{\mathrm{site}}(K,\gate,\tmpl;\delta):\quad
  \sup_{\bar c\in\overline{\Ctx}}\ \sup_{s_0,s_1}\
  d\bigl((O_{\tmpl})_{\#}\mu^{\gate}_{K,\bar c[s\leftarrow s_0]},\;
         (O_{\tmpl})_{\#}\mu^{\gate}_{K,\bar c[s\leftarrow s_1]}\bigr)
  \ \leq\ \delta.
  \label{eq:phi-site}
\end{equation}
Total variation contracts under a measurable map, so $\Phi_{\mathrm{NI}}$ forces
$\Phi_{\mathrm{site}}$ everywhere, and relaxing a violating site already forfeits
the global criterion. A \emph{gate} sends admission-time evidence about $K$ to
$\{\mathsf{mediate},\mathsf{relax}\}$, and is \emph{$\delta$-sound} when
\Cref{eq:phi-site} holds at every site for every artifact in its declared domain,
a class the always-mediate gate shows is nonempty. Write $\gV(K,\tmpl)$ for the
variants in that domain sharing code, interface, and site but carrying different
attacker-written specifications; \Cref{app:proof-theorem1} fixes the reference
profile and the conservative transformation it applies.

Admission decisions may use three kinds of artifact evidence: facts established
over code, as in proof-carrying code \citep{necula1997}; a sound abstraction of the artifact against its
interface \citep{sekar2003} or of agent behavior reconstructed from traces
\citep{agentarmor}; and dependencies read off the specification, which is what a
registry reviewer must do when an artifact ships as prose
\citep{mcptox,mcpsecbench,skillattack}. Two conditions on such a route carry the
argument. C1 says the route assigns one judgment to a safe and a leaking member
of $\gV(K,\tmpl)$. C2 says the attacker can exhibit a context in which the
leaking member violates \Cref{eq:phi-site}. Neither is a claim about every
admission system; \Cref{app:proof-theorem1} states both formally.
\begin{theorem}[No expansion under judgment ambiguity]
\label{thm:empty-unsound}
Fix $\delta$, a site $\tmpl$, a route $i$, and a profile $\gate_0$, and suppose
C1 and C2 hold at $(K,\tmpl)$. Let $g$ be a $\delta$-sound gate whose decision
depends on $K$ only through $\mathcal{J}_i(K,\tmpl)$ and whose induced mediator
agrees with $\gate_0$ off $\tmpl$ at every artifact in its domain. Then $g$
returns $\mathsf{mediate}$ for both variants of C1.
\end{theorem}
\vspace{-5pt}
Equal judgments force equal decisions, and C2's context then makes the dependent
variant violate \Cref{eq:phi-site}. The work sits in C1, which for code and
code-interface evidence reduces to the variant class holding one safe and one
leaking member, and for text extraction to a semantics-changing rewrite the
extractor cannot see (\Cref{app:proof-theorem1}). Evidence that separates the
variants is untouched, including the code-resolved case
(\Cref{prop:code-resolved}). The theorem also locates where the missing evidence
first appears. Immediately before a tool runs, the concrete arguments, the
resources already read, the authenticated request, and a hook that can still
change the effect coexist, none of it available at admission time. Reachability
alone is not an authorization decision there, since a represented path says only
that an effect may occur and cannot tell a requested, taint-free call from an
attacker-induced one. The boundary therefore needs two checks, each landing on an
action the runtime implements. The next section therefore has three requirements:
represented unsafe paths must be severable, allowed effects must be authorized,
and every abstract verdict must correspond to a real intervention.

\section{\method{}: Provenance-Aware Capability
Enforcement}
\label{sec:method}
\begin{figure}[t]
    \centering
    \includegraphics[width=0.76\linewidth]{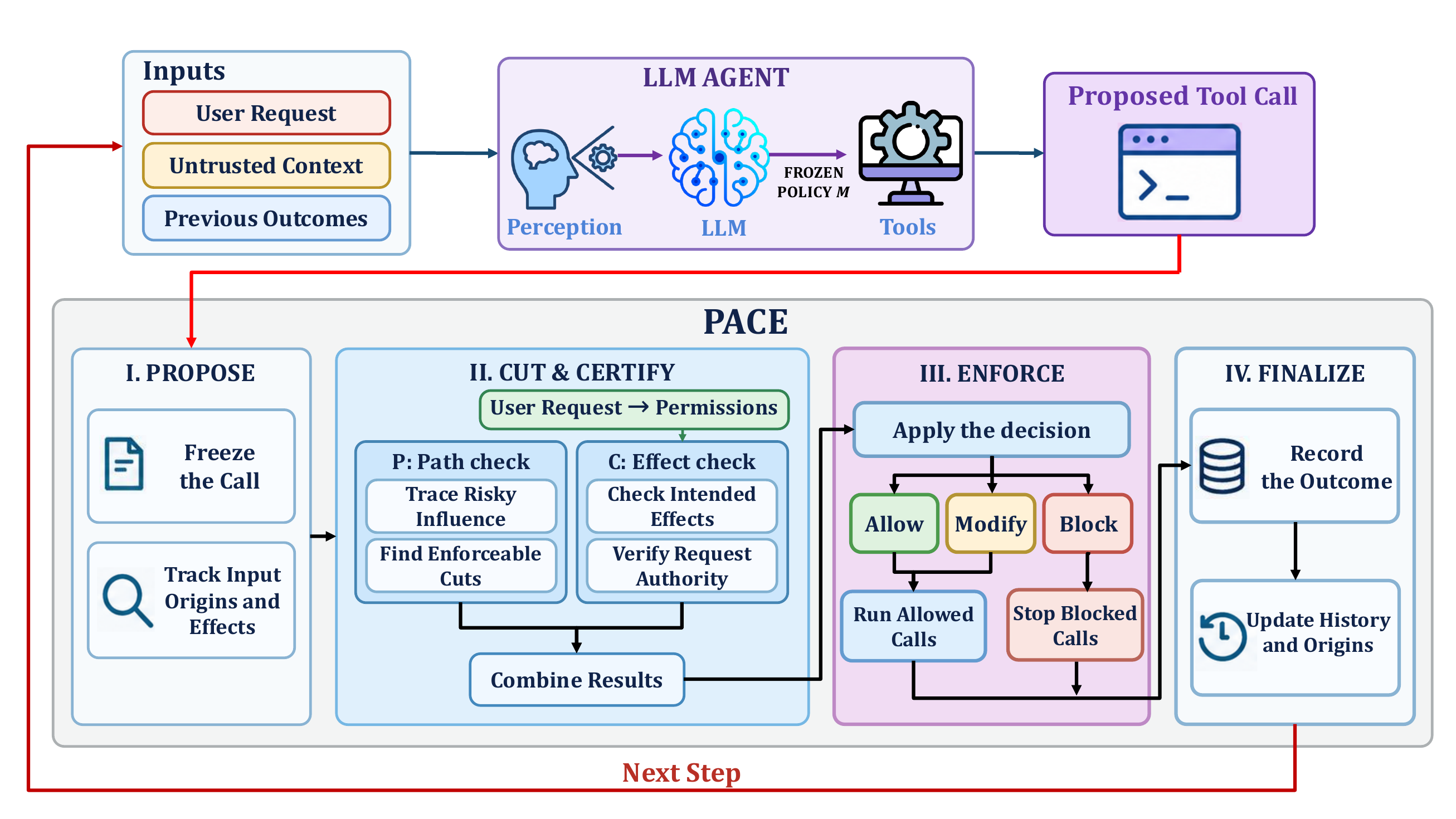}
    \caption{One mediated step of \method{}. Phase \Rmnum{1}--PROPOSE (\cref{sec:graph}) freezes the call, extends the graph, and
  expands the effect atoms; phase \Rmnum{2}--CUT \& CERTIFY (\cref{sec:cut}) cuts and certifies against committed evidence
  alone; phase~\Rmnum{3}--ENFORCE (\cref{sec:enforce-phases}) installs the manifest and dispatches; phase \Rmnum{4}--FINALIZE (\cref{sec:enforce-phases}) commits the
  outcome. }
    \label{fig:system-overview}
    \vspace{-5pt}
\end{figure}

As shown in \Cref{fig:system-overview}, \method{} handles each mediated call
in four phases. \emph{Propose} freezes its arguments, extends the episode
provenance graph, and expands schema-defined effects. \emph{Cut and Certify}
proposes a cut of represented influence paths and checks effects against
request-derived capabilities; verifiers recompute both certificates.
\emph{Enforce} installs the selected action before dispatch, and
\emph{Finalize} records the outcome for the next call. We distinguish two
configurations. \method{}-\textsc{c} is the certified contract: the installed
action preserves the cut and the final call passes effect verification.
All experiments evaluate \method{}-\textsc{p}, whose combiner can restore
C-certified calls blocked by P and dispatch declared repairs without final-call
recertification. Its fast paths and the conditions for retaining the certified
guarantee are specified below.

\subsection{Phase \Rmnum{1}: PROPOSE}
\label{sec:graph}
Phase \Rmnum{1} constructs the evidence for each boundary decision. Dynamic
tainting would need every transformation to carry a label
\citep{taintdroid,libdft}, which tool-using agents defeat by crossing opaque
boundaries. \method{} records declared structure instead. Each proposal yields a
path view over the episode graph and an effect view over the effects the schema
declares. Both are folded from an append-only audit record, so a verifier can
recheck a decision without rerunning the agent.

\begin{definition}[Evidence and Active Graphs]
\label{def:graph}
The \emph{evidence graph} $H_t=(V_t,E_t)$ is a finite directed multigraph over
data, operation, resource, and terminal nodes, folded from the episode record.
Proposals, outcomes, blocks, cancellations, and uncertain dispatches are appended
and never rewritten, so the committed part of $H_t$ only grows within an
episode. Reachability uses only the flow subgraph
\begin{equation}
\begin{aligned}
  E^{\mathrm{fl}}_t=\bigl\{\,e\in E_t\ :\ \mathrm{type}(e)\in\{\,
  &\mathtt{resource\_access},\ \mathtt{data\_flow},\ \mathtt{output\_flow},\\
  &\mathtt{resource\_write},\ \mathtt{terminal}\,\}\,\bigr\}.
\end{aligned}
\label{eq:flow}
\end{equation}
The \emph{active graph} $G^{-}_t$ is what one decision reads: the committed edges
of $H_t$ at $\sigma_t$ together with the proposed edges of $a_t$, excluding
refused and cancelled proposals. Two counters index that history. The execution version $\sigma_t$ advances on any
change to committed edges, labels, or live values. The policy epoch
$\varepsilon_t$ advances on any change to schemas, effect rules, costs,
contracts, capability grants, or adapter code. Each record advancing $\sigma_t$
stores the live-value set with an origin label per entry.
\end{definition}

\tightparagraph{Snapshot and Provenance.}
Against that state, the host first freezes the call under a fresh attempt
identifier $q_t$, storing each argument by value or as a versioned resource
handle. Later phases read the snapshot, never a value the agent controls. A
resource that is neither snapshottable nor versioned is schema-undefined, so rule
R1 refuses the call. The host then matches each frozen argument against a
provenance map $\nu$ of live values. A hit reuses the existing node, linking this
call to committed history and to any protected value behind it. A miss mints a
new data node. Provenance is kept per slot, so arguments, memory records,
artifacts, and output fields carry separate origins.

\tightparagraph{Emission and Authority.}
With every argument now a node, the host wires the proposal into the path view.
One edge runs from each schema-defined argument node to a fresh operation node
$\mathrm{op}(q_t)$. A fresh terminal $T(q_t)$ closes it, and the versioned
unsafe-effect rule labels $\mathrm{op}(q_t)\to T(q_t)$. Only these
\emph{proposed} edges remain severable by the executor. A position with no schema
gets no edge, and an operation outside the abstraction becomes an unknown
transition, wired to an unsafe terminal if it may be irreversible. The effect
view expands the same frozen call through the effect schema into \emph{effect
atoms}. Each names an operation, a target resource, a persistence class, and the
argument slots that determine it. The authenticated request compiles separately
into a capability bounded by purpose, resource, arguments, budget, expiry, and
delegation depth. It reads only the trusted request channel, so model-generated
text can neither mint nor widen a capability.

Phase \Rmnum{2} reads $G^{-}_t$, the effect atoms, the capability set, and the pinned
counters without modifying the frozen decision state.

\subsection{Phase \Rmnum{2}: Cut and Certify}
\label{sec:cut}

Phase \Rmnum{2} runs two checks over that frozen material. The path chain decides
which proposed edges to sever and binds each to an executor action; the effect
chain decides whether the concrete effects are authorized. Their certificates answer different questions; whether both guarantees survive
depends on the action the combiner actually installs.

\tightparagraph{Sources and the unsafe terminal.}
The path chain starts from two source sets. $P_t$ holds only \texttt{Data}
and \texttt{Resource} nodes labeled as secrets, personal data, credentials,
source code, or protected payloads. A derivation check also scans proposals and
observations for registered values over a finite family $\Psi$ of deterministic
re-encodings, each canonicalizing in its own alphabet:
\begin{equation}
  \mathrm{match}(y,x)=1
  \iff \exists\psi\in\Psi:\ \psi(x)\ \text{occurs in}\
  \mathrm{canon}_{\psi}(y).
  \label{eq:oracle}
\end{equation}
A match adds that argument's data node to $P_t$ for this step, as a label and not
an edge, so encodings are caught without treating substring matching as taint.
The check is a bounded approximation: paraphrase, translation, and fragmentation
lie outside $\Psi$, and \Cref{app:scope} lists the misses and false positives.
The other set, $B_t$, over-approximates where attacker influence may sit. Every
committed node carries a deployment-assigned origin, untrusted covers anything an
artifact, tool return, retrieved page, or memory record supplied, and a derived
node inherits its inputs'. The verifier closes the untrusted-origin live values
at $\sigma_t$ over committed flow edges only (\Cref{app:coverage}), so a call
built from the user request and the host alone needs no cut. At the far end, the
registry runs $\mathrm{unsafe}_{\varepsilon_t}(a_t)$ on the canonical proposal,
so $U_t=\{T(q_t)\}\cap\terminalbad$ holds the terminal of an unsafe call and
nothing else. It is disjoint from $P_t$, so the cut problem is well posed and a
safe call takes the empty cut.

\tightparagraph{Contracts and binding.}
With the endpoints fixed, binding runs before optimization, so the cut ranges
only over edges some adapter can sever. For each proposed edge $e$ the trusted
adapter declares the actions it implements, $\Impl(e)$, and the deployment
declares the subset $\Sev(e)$ whose warrants rule out every real link through
$e$ once the action is installed. \Cref{app:deferred} states those warrants and
shows why they leave redaction rarely licensed for model-chosen calls, so an
argument edge with no severing action pushes the cut back to the terminal, which
admits \texttt{block\_call}. Both lists are finite and canonically ordered, so
the binding takes the least declared cost among the survivors,
\begin{equation}
  \bind_t(e)\in\arg\min_{a\in\Sev(e)\cap\Impl(e)}\mathrm{cost}(a,e),
  \qquad \bind_t(e)=\bot\ \text{if the intersection is empty.}
  \label{eq:binding}
\end{equation}
\tightparagraph{The cut and the manifest.}
Those bindings become capacities on the flow subgraph of $G^{-}_t$:
\begin{equation}
  \kappa(e)=
  \begin{cases}
    \mathrm{cost}\bigl(\bind_t(e),e\bigr), &
      \bind_t(e)\neq\bot\ \text{and}\ e\ \text{is \texttt{proposed} for}\ a_t,\\
    \infty, & \text{otherwise},
  \end{cases}
  \qquad e\in E^{\mathrm{fl}}_t.
  \label{eq:capacity}
\end{equation}
A super-source joins $B_t\cup P_t$ and a super-sink follows $U_t$, and $\cut_t$
is the canonical minimum cut \citep{fordfulkerson}: cheapest, lexicographically
least at equal cost. Committed edges are therefore uncuttable, since the present
executor cannot change a past call. Control and intent edges never enter the
network, and strictly positive rational costs bar padding with free edges. Under
the emission rules the optimum reduces to a two-way choice that one reachability
pass settles in linear time (\Cref{prop:cut-closed-form}).
What the executor consumes is not that cut but the edge-indexed \emph{manifest}
$\manifest_t=\{(e,\bind_t(e),\mathrm{arg}(e),\mathrm{hook}(e)):e\in\cut_t\}$,
where $\mathrm{arg}(e)$ and $\mathrm{hook}(e)$ name the registry-bound argument
position and hook, one per cut-edge instance. It is \emph{admissible} only if its
entries do not conflict, every hook runs in time, and the deployment declares
that each severance contract survives under the complete set (\Cref{app:verify}).
Only that declaration is checked here, and its truth is a premise no automatic
check reaches. A conflict triggers refusal, not a search for a costlier cut.

\tightparagraph{The effect chain and the decision rule.}
The effect chain checks nine obligations over every schema-defined call:
authorization coverage, typed binding, operation and principal control,
confidentiality, persistence, delegation attenuation, budget, freshness against
request and event state, and hook availability. Write $\mathrm{eff}(a)=1$ when
all atoms of $a$ pass. For a nonblocking manifest, let $a'_t$ be its resulting call,
and $\mathrm{low}(a_t)=1$ when every effect is schema-labelled reversible, has
an empty sink, and has no persistent side effect. The influence test
$\mathrm{infl}(P_t,a_t)$ is one when a protected node reaches
$\mathrm{op}(a_t)$ or the derivation check matches a protected form;
$\mathrm{auth}(a_t)$ comes from a channel the artifact cannot write.
The certified decision contract uses the following ordered rules with R2 and R3
disabled; the evaluated fast paths are shown explicitly:
\begin{equation}
  \mathrm{dec}(a_t)=
  \begin{cases}
    \texttt{refuse} & \text{\textbf{R1} } \mathrm{schema}(\mathrm{op}(a_t),j)
                     =\bot\ \text{for the operation or some argument}\ j,\\
    \texttt{allow} & \text{\textbf{R2} } \mathrm{low}(a_t)=1,\\
    \texttt{allow} & \text{\textbf{R3} } \mathrm{auth}(a_t)=1
                     \wedge\mathrm{infl}(P_t,a_t)=0\wedge\mathrm{eff}(a_t)=1,\\
    \texttt{refuse} & \text{\textbf{R4} } \mathrm{eff}(a_t)=0,\\
    \texttt{refuse} & \text{\textbf{R5} } \text{no finite cut, or}\
                     \manifest_t\ \text{inadmissible},\\
    \texttt{refuse} & \text{\textbf{R6} } \manifest_t\ \text{nonblocking}
                     \wedge\mathrm{eff}(a'_t)=0,\\
    \manifest_t & \text{otherwise, with}\ \cut_t\ \text{possibly empty.}
  \end{cases}
  \label{eq:decision}
\end{equation}
R1 rejects undefined schemas. In \method{}-\textsc{c}, R4 checks the original
call, R5 requires an executable cut, and R6 checks the final effects of a
nonblocking manifest; a terminal block dispatches no call. In
\method{}-\textsc{p}, R3 requires all nine obligations to pass, whereas R2
records their results and can release a low-impact call despite failure of AUTH
or BUDGET. Beyond these fast paths, its combiner can restore the original call
when P selects a terminal block and C certifies it, or apply a registry-declared
repair to a repairable verdict without rerunning R6. Thus this equation specifies
the certified contract, not a claim that every evaluated dispatch satisfies it.
With sound request compilation, effect labels, validators, and complete mediation, an
out-of-authority irreversible or outward effect can execute only through a
declared repair: R2 excludes it and unrepaired release requires C's acceptance.
\Cref{app:scope,app:proof-episode} give the premises and argument.

\tightparagraph{Two certificates.}
Following proof-carrying code \citep{necula1997,sekar2003}, \method{} ships that
candidate decision as evidence rather than as a trusted verdict. The path certificate
$\Cert_t=(q_t,L^{\mathrm{p}}_t,\sigma_t,\varepsilon_t,B_t\cup P_t,U_t,\cut_t,
\manifest_t)$ carries the state it was computed against and every input to the
decision. Its verifier rebuilds the graph, the sources, and the terminal from the
episode record and the versioned registry, then rechecks separation and every
binding; \Cref{app:verify} lists the ten conditions in implementation order. The
effect certificate binds the atoms to the request, the schema, the capability
witnesses, and the provenance slice a check needs, and a deterministic verifier
reruns the nine obligations over it. The generator solves the cut and the
verifier checks the one it is handed, so we claim recomputed separation and
enforceability, not recomputed optimality. Acceptance certifies the candidate;
it does not show that the combiner executed its manifest.

\subsection{Phases \Rmnum{3} and \Rmnum{4}: Enforce and Finalize}
\label{sec:enforce-phases}

\begin{definition}[Grounded action]
\label{def:grounded}
A cut edge's bound action is \emph{grounded} at step $t$ when the adapter exposes
the corresponding hook, the executor installs it before the represented
dependency can be traversed, and its severance contract holds under the complete
manifest, not in isolation. A manifest is grounded when all entries are. An
ungrounded entry makes the certified branch refuse or cancel.
\end{definition}

\tightparagraph{Enforce.}
The executor consumes the combined decision before it invokes the tool. Three
actions suppress an effect, blocking, refusal, and a dry run that commits
nothing; four constrain one, a read-only view, an argument rewrite, sanitization
at a mediated boundary, and a budget cap; and allow passes the call through. A
severance contract may name only these, which places the mediator among the
editing enforcement mechanisms rather than the truncating ones
\citep{ligatti2005,hamlen2006}. In the certified contract, installation precedes dispatch and the call carries
exactly the manifest-produced bytes, with R6 checked on those bytes. Restored or
repaired calls in \method{}-\textsc{p} follow the combiner's selected action.
A recheck of $\sigma_t$ and $\varepsilon_t$ cancels a stale attempt. That freshness is process-local. Durable reservation,
one-shot token spending, and crash reconciliation are deployment hardening,
specified in \Cref{app:protocol} and outside the artifact we evaluate.

\tightparagraph{Finalize.}
After execution \method{} records a \texttt{ToolEvent}: whether the action ran,
the resolved resource, the observable output, artifact and network effects, the
executor mode, and the guard's reason. A blocked or dry-run call is recorded as
not executed, an allowed or rewritten one as the effect observed. That event
updates the graph, the provenance map, and the taint state the next proposal
reads. A confirmed outcome commits its output edges and mints a fresh provenance
identifier per returned value, unless the adapter certifies a pass-through; a
dispatch with no observed output commits the proposal alone, since the call may
still have taken effect. Either way the returned value first appears at step
$t+1$ and never enters its own certificate. The final response crosses the same
boundary under the same state, and its terminal is registered unsafe whenever the
influence test fires.

\subsection{Security Guarantees}
\label{sec:guarantee}

The certificates concern represented-path separation and authority over
schema-exposed effect atoms; executor grounding connects cuts to real links. Neither establishes
noninterference. A \method{}-\textsc{p} step inherits the guarantees below only
when its final action meets their premises (proofs in \Cref{app:proofs}).

\begin{proposition}[Certificate soundness]
\label{thm:cert-separation}\label{prop:effect-sound}
\emph{(a) Represented-path separation.} If the path verifier accepts $\Cert_t$,
removing $\cut_t$ makes $U_t$ unreachable from $B_t\cup P_t$ in $G^{-}_t$.
Each cut edge is a proposed flow-edge instance of $q_t$ with a re-derived
severing action; $\manifest_t$ binds each exactly once and is \emph{declared}
jointly admissible; $\kappa(\cut_t)$ equals the stated cost; and the named
state and epoch are current at acceptance.
\emph{(b) Authorized-effect soundness.} Assume capabilities issue only from
authenticated requests, the effect schema and slot provenance cover the concrete
effect, declared validators and declassifiers meet their contracts, and every
required hook is applied. If the effect verifier certifies $a_t$, every effect
atom is covered by a valid capability and satisfies the nine obligations
of \Cref{sec:cut}.
\end{proposition}

Minimality remains declared, and omissions shared with the builder remain
possible (\Cref{app:scope}). Rewrites can change effects, so
\method{}-\textsc{c} rechecks (b) on $a'_t$ under R6.

For a real causal chain $\gamma=x_0\to\cdots\to x_{m-1}\to\zeta_t$ ending at
the pending effect, a link is \emph{covered} if its adapter schema exposes the
relevant value position and a timely hook, and \emph{emitted} if its flow edges
appear before dispatch. Under maps $\alpha_{\mathrm D}$ to data nodes and
$\alpha_{\mathrm T}$ to a terminal, $\gamma$ is \emph{simulated} if its
endpoints match, its per-link segments concatenate into one path in
$E^{\mathrm{fl}}_t$, and each real link requires traversing every edge of its
segment.

\begin{theorem}[Conditional grounded execution]
\label{thm:confinement}
Let both verifiers accept with grounded $\manifest_t$ (\Cref{def:grounded}).
Let $\gamma$ connect a registered protected source to the pending unsafe effect,
with every link covered, emitted in time, and simulated, and complete source and
terminal labels (H3a and H3c of \Cref{app:proofs}). If the proposal, state, and
epoch remain current through dispatch, $\gamma$ cannot complete at this call.
Replacing H3a with H3b gives integrity (\Cref{cor:integrity}).
\end{theorem}

Simulation maps $\gamma$ to a $P_t$--$U_t$ path crossing the cut. Grounded
severance prevents the corresponding real link; freshness aligns certified and
dispatched proposals. \Cref{prop:transparent} gives sufficient adapter conditions
for emission and simulation.

\begin{proposition}[Transparent chains are simulated]
\label{prop:transparent}
Let $\gamma$ have $m\geq1$ and start at a registered value with a host provenance
identifier. Suppose each call reads, writes, and returns only at declared
positions, issues no hidden secondary effect, and preserves pass-through value
identifiers; each $x_{i+1}$ reaches the next call as a host-held handle or bytes
$\nu$ matches to its own node; and every earlier call producing a value or
resource version in $\gamma$ was finalized before $q_t$ was proposed. Under the
host emission rules of \Cref{sec:graph,app:coverage}, $\gamma$ is emitted and
simulated in $G^{-}_t$.
\end{proposition}

The transparent fragment covers typed handles, host-managed values,
pass-through chains, and host-versioned resources, including internal computation
(\Cref{rem:transparent-scope}). It excludes model-recomposed strings without
host provenance, unexposed effects, tool-internal secondary calls, and writes
bypassing versioned adapters. Observable refusal remains a channel, with no bound of the form
\Cref{eq:phi-site} proved at any $\delta$.
Final actions must preserve separation: in \Cref{prop:cut-closed-form}, removing
$S\subseteq A\cup\{e_{\mathrm T}\}$ separates sources from the unsafe terminal
exactly when $e_{\mathrm T}\in S$ or $A^\star\subseteq S$, by each unsafe
path's last two edges (\Cref{app:proof-cut}). Restoring a blocked call removes
neither, losing separation when $A^\star\ne\emptyset$ even if C authorizes it.
Complete mediation and the remaining theorem premises extend the per-call result
to an episode invariant for \method{}-\textsc{c} (\Cref{cor:episode}), with no
such general claim for the evaluated configuration.

\begin{table}[tbp]
\centering
\caption{Main results on benchmarks with target-model families. ASR
is the worst attack group for that benchmark and model. Each block
shows the undefended agent, \method{}, and the three strongest published
defenses, ranked by mean worst-group ASR. Util.\ is the benchmark's own score and
is higher better except on MCPTox, where it is the refusal rate. Best ASR per block and model in bold.}
\label{tab:main}
\small
\setlength{\tabcolsep}{2.5pt}
\begin{minipage}[t]{0.5\linewidth}\centering
\scalebox{0.70}{
\begin{tabular}{@{}lcccccc@{}}
\toprule
& \multicolumn{2}{c}{\textbf{Qwen}} & \multicolumn{2}{c}{\textbf{DeepSeek}} & \multicolumn{2}{c}{\textbf{ChatGPT}}\\
\cmidrule(lr){2-3}\cmidrule(lr){4-5}\cmidrule(lr){6-7}
Method & ASR $\downarrow$ & Util. & ASR $\downarrow$ & Util. & ASR $\downarrow$ & Util. \\
\midrule
\multicolumn{7}{@{}l}{\emph{AgentDojo}\ \ }\\[1pt]
Baseline & 48.0 & 62.59 & 99.84 & 75.68 & 100.0 & 73.77 \\
CaMel & 0.32 & 35.14 & 28.14 & 60.04 & 36.09 & 57.02 \\
DTA & 0.95 & 29.84 & 33.33 & 54.0 & 55.01 & 56.01 \\
DataSent. & 35.29 & 73.77 & 78.06 & 85.0 & 84.1 & 82.95 \\
\method{} & \textbf{0.0} & 59.67 & \textbf{5.09} & 75.25 & \textbf{8.11} & 72.05 \\
\midrule
\multicolumn{7}{@{}l}{\emph{AgentDyn}\ \ }\\[1pt]
Baseline & 32.5 & 63.39 & 19.5 & 65.0 & 19.5 & 74.82 \\
Repeat Prompt & 34.5 & 61.25 & 31.5 & 61.43 & 31.5 & 70.54 \\
CaMel & \textbf{0.0} & 0.0 & 5.5 & 3.75 & \textbf{0.0} & 0.0 \\
DRIFT & 0.56 & 29.82 & 2.78 & 13.39 & 4.5 & 14.64 \\
\method{} & \textbf{0.0} & 60.54 & \textbf{0.06} & 63.57 & \textbf{0.0} & 73.57 \\
\midrule
\multicolumn{7}{@{}l}{\emph{ASB}\ \ }\\[1pt]
Baseline & 88.5 & 30.29 & 74.08 & 64.0 & 81.08 & 80.25 \\
Sandwich & 88.5 & 32.5 & 81.83 & 69.07 & 71.21 & 87.86 \\
Direct Para. & 71.5 & 62.0 & 79.0 & 81.64 & 81.75 & 81.57 \\
PoT Shuffle & 99.0 & 28.57 & 65.42 & 57.43 & 59.25 & 67.04 \\
\method{} & \textbf{0.42} & 29.57 & \textbf{10.5} & 64.64 & \textbf{8.75} & 81.63 \\
\midrule
\multicolumn{7}{@{}l}{\emph{WASP}\ \ }\\[1pt]
Baseline & 19.05 & 78.57 & 9.52 & 79.76 & 4.76 & 86.9 \\
Prompt Filter & 14.29 & 82.14 & 7.14 & 75.0 & 11.9 & 83.33 \\
Simple Static & 20.37 & 72.64 & 16.67 & 64.29 & 19.04 & 84.52 \\
FIDES & 40.48 & 0.0 & 30.95 & 32.14 & 23.81 & 59.52 \\
\method{} & \textbf{7.14} & 92.86 & \textbf{0.0} & 91.67 & \textbf{0.0} & 95.24 \\
\bottomrule
\end{tabular}}
\end{minipage}%
\begin{minipage}[t]{0.5\linewidth}\centering
\scalebox{0.70}{
\begin{tabular}{@{}lcccccc@{}}
\toprule
& \multicolumn{2}{c}{\textbf{Qwen}} & \multicolumn{2}{c}{\textbf{DeepSeek}} & \multicolumn{2}{c}{\textbf{ChatGPT}}\\
\cmidrule(lr){2-3}\cmidrule(lr){4-5}\cmidrule(lr){6-7}
Method & ASR $\downarrow$ & Util. & ASR $\downarrow$ & Util. & ASR $\downarrow$ & Util. \\
\midrule
\multicolumn{7}{@{}l}{\emph{InjecAgent-E}\ \ (no util.)}\\[1pt]
Baseline & 79.07 & -- & 52.29 & -- & 50.5 & -- \\
Sandwich & 0.59 & -- & 50.62 & -- & 73.12 & -- \\
Instr. Prevention & 58.82 & -- & 67.29 & -- & 56.63 & -- \\
PI-Det. & 20.0 & -- & 56.08 & -- & 57.36 & -- \\
\method{} & \textbf{0.0} & -- & \textbf{3.12} & -- & \textbf{30.77} & -- \\
\midrule
\multicolumn{7}{@{}l}{\emph{MCPTox}\ \ (Util.: refusal $\downarrow$)}\\[1pt]
Baseline & 32.52 & 0.16 & 56.54 & 0.42 & 86.93 & 27.96 \\
Allowlist & 0.72 & 2.46 & 15.18 & 7.66 & 16.29 & 55.78 \\
MCIP & 0.51 & 54.23 & 9.95 & 76.72 & 10.04 & 81.22 \\
Progent & 8.8 & 45.88 & 20.42 & 66.85 & 46.4 & 78.43 \\
\method{} & \textbf{0.0} & 1.97 & \textbf{2.62} & 0.58 & \textbf{2.65} & 28.55 \\
\midrule
\multicolumn{7}{@{}l}{\emph{PASB}\ \ (IPI util.)}\\[1pt]
Baseline & 85.0 & 98.47 & 92.5 & 99.24 & 97.5 & 100.0 \\
Delimiters & 45.0 & 99.23 & 55.0 & 98.47 & 35.0 & 99.24 \\
Instr. Prevention & 17.5 & 97.67 & 27.5 & 96.95 & 27.5 & 98.47 \\
Progent & 80.0 & 100.0 & 77.5 & 100.0 & 62.5 & 98.47 \\
\method{} & \textbf{0.0} & 97.69 & \textbf{5.0} & 99.24 & \textbf{2.5} & 100.0 \\
\midrule
\multicolumn{7}{@{}l}{\emph{MSB}\ \ (Util.: PUA)}\\[1pt]
Baseline & 38.75 & 46.34 & 100.0 & 91.24 & 98.75 & 89.08 \\
MCIP & 50.0 & 59.15 & 77.5 & 96.55 & 76.25 & 92.34 \\
PI-Guard & 37.5 & 56.42 & 57.5 & 91.98 & 66.25 & 89.97 \\
Tool Filter & 30.81 & 47.34 & 59.35 & 93.65 & 52.5 & 91.04 \\
\method{} & \textbf{2.5} & 60.14 & \textbf{3.75} & 98.83 & \textbf{3.75} & 94.87 \\
\bottomrule
\end{tabular}}
\end{minipage}
\end{table}
\section{Experiments}
\label{sec:experiments}

In this section, we give the experimental settings and show some compared results with more detailed in \cref{app:full-tables}. All experiments evaluate \method{}-\textsc{p}, denoted \method{} in the tables.

\subsection{Experimental Setup}
\label{sec:setup}
\tightparagraph{Models and measurement.}
Three frozen target families, \texttt{gpt-5.6-luna},
\texttt{DeepSeek-V4-Flash}, and local vLLM-served \texttt{Qwen-3.8-27B},
decode deterministically at temperature 0 with a 1024-token completion limit. Target, attack source, and judge roles
are separate; a model from the target's family never judges alone. Native metric
point estimates are unpooled across benchmarks. Rates record numerators, valid
denominators, invalid and technical-failure counts; missing required defense-hook
invocations invalidate defended episodes
(More information can be found in \Cref{app:hyperparameters,app:benchmark-details,app:statistics}).

\tightparagraph{Benchmarks and baselines.}
Eight executable benchmarks retain native attack construction: AgentDojo
\citep{agentdojo}, AgentDyn \citep{agentdyn}, WASP \citep{wasp}, InjecAgent
\citep{injecagent}, ASB \citep{asb}, PASB \citep{pasb}, MCPTox
\citep{mcptox}, and MSB \citep{msb}. Baselines span five families:
trusted flow and information-flow control (CaMeL \citep{camel}, FIDES
\citep{fides}, DRIFT \citep{drift}); tool-use safeguards (Progent
\citep{progent}, MCIP Guardian \citep{mcip}, ToolShield \citep{toolshield},
AgentDojo tool filter, MCPTox allowlist); re-execution (MELON \citep{melon});
detection (DataSentinel \citep{datasentinel}, PIGuard \citep{piguard},
PromptGuard-2 \citep{promptguard2}, AgentDojo injection detector, MCPTox metadata
sanitizer); and ASB and WASP prompt-level defenses, including spotlighting
\citep{hines2024spotlighting}. Fully instrumented DTA \citep{taintdroid,libdft}
is a diagnostic taint-tracking reference excluded from main-table comparator selection
(\Cref{app:port-fidelity}). Comparisons run locally at pinned revisions.
Capabilities compile deterministically from the trusted request channel, excluding attack
labels and judge outputs. Benchmark-specific policies and schemas were frozen
before final runs (\Cref{app:benchmark-details}).

\subsection{Main Results}
\label{sec:main-results}

Within benchmarks, \Cref{tab:main} compare the worst eligible
attack group per model; \Cref{app:full-tables} gives all groups and the
comparator-selection rule. Among 79 eligible per-group ASR columns,
\method{}-\textsc{p} is strictly lowest in 62, tied lowest in 14, and higher in
three, excluding the derived InjecAgent data-stealing aggregate, its three
conditional DS-S2 columns, and two columns with zero undefended ASR. On \texttt{Qwen-3.8-27B}, four of five ASB groups have zero ASR. Residuals
are 4.76\% and 7.14\% on WASP, 0.42\% on ASB-DPI, and 2.50\% on MSB-Call.
On AgentDojo, \method{} attains lower worst-group ASR and higher utility than
CaMeL, MELON, and Progent on all three models. Relative to undefended agents,
full-benchmark utility drops at most 2.92 points, gains 14.29, 11.91, and
8.34 points on WASP, and improves in MSB PUA and NRP; MCPTox refusal differs
by at most 1.81 points. The three lower-ASR comparator columns have much lower
companion utility (\Cref{app:full-tables}). Reduced-scale adaptive search
on 30 out-of-authority AgentDojo pairs yields 21/30 (70.0\%) undefended
success, including nine first successes from search-generated payloads,
versus 0/30 for \method{}-\textsc{p} (two-sided 95\% Wilson interval
[0, 11.4\%]). All 157 target calls at its boundary were blocked before
dispatch without R2 release. \Cref{app:benchmark-details,app:additional-results}
detail the one-model, three-suite protocol and controlled checks.
\begin{wraptable}[17]{r}{0.45\textwidth}
\centering
\caption{Ablation study with the local \texttt{Qwen-3.8-27B} endpoint. }
\label{tab:ablation}
\small
\setlength{\tabcolsep}{2.5pt}
\resizebox{\linewidth}{!}{
\begin{tabular}{@{}lcccccccc@{}}
\toprule
& A0 & A1 & A2 & A3 & A4 & A5 & A6 & A7 \\
\cmidrule(lr){2-9}
Path confinement (P)            & & \checkmark & & & \checkmark & \checkmark & & \checkmark \\
Capability and effect (C)       & & & \checkmark & & \checkmark & & \checkmark & \checkmark \\
Boundary adaptation (B)         & & & & \checkmark & & \checkmark & \checkmark & \checkmark \\
\midrule
\multicolumn{9}{@{}l}{\emph{Attack success (\%), lower better}}\\[1pt]
WASP, plain text                & 8.3 & 8.3 & 0.0 & 0.0 & 8.3 & 0.0 & 8.3 & \textbf{0.0} \\
WASP, URL injection             & 8.3 & 8.3 & 8.3 & 0.0 & 0.0 & 0.0 & 16.7 & \textbf{0.0} \\
AgentDojo, overall              & 37.4 & 0.0 & 0.0 & 33.3 & 0.0 & 0.0 & 0.0 & \textbf{0.0} \\
InjecAgent-Base, valid          & 14.7 & 0.0 & 0.0 & 14.6 & 0.0 & 1.0 & 0.0 & \textbf{0.0} \\
InjecAgent-Enhanced, valid      & 0.0 & 0.0 & 0.0 & 0.0 & 0.0 & 0.0 & 0.0 & \textbf{0.0} \\
ASB, overall                    & 0.0 & 0.0 & 0.0 & 0.0 & 0.0 & 0.0 & 0.0 & \textbf{0.0} \\
AgentDyn, overall               & 25.0 & 0.0 & 1.7 & 23.3 & 1.7 & 0.0 & 0.0 & \textbf{0.0} \\
PASB, overall                   & 41.4 & 31.0 & 0.0 & 41.4 & 0.0 & 39.7 & 0.0 & \textbf{0.0} \\
MCPTox, valid                   & 5.2 & 0.8 & 0.0 & 5.4 & 0.0 & 0.8 & 0.0 & \textbf{0.0} \\
MSB, official macro             & 24.8 & 7.1 & 4.5 & 0.0 & 5.8 & 0.0 & 0.0 & \textbf{0.0} \\
\midrule
\multicolumn{9}{@{}l}{\emph{Native utility and refusal (\%)}}\\[1pt]
WASP, utility $\uparrow$        & 83.3 & 83.3 & 83.3 & 95.8 & 79.2 & 66.7 & 87.5 & 83.3 \\
AgentDojo, utility $\uparrow$   & 76.8 & 13.6 & 64.6 & 76.8 & 65.2 & 13.6 & 66.2 & 67.2 \\
AgentDyn, utility $\uparrow$    & 66.7 & 0.0 & 36.7 & 66.7 & 36.7 & 0.0 & 31.7 & 33.3 \\
ASB, utility $\uparrow$         & 0.0 & 0.0 & 0.0 & 0.0 & 30.0 & 0.0 & 0.0 & 30.0 \\
ASB, refusal $\downarrow$       & 39.3 & 41.8 & 38.6 & 39.6 & 37.1 & 43.2 & 41.8 & 40.7 \\
PASB, IPI utility $\uparrow$    & 100.0 & 0.0 & 100.0 & 100.0 & 100.0 & 0.0 & 100.0 & 100.0 \\
MCPTox, refusal $\downarrow$    & 42.2 & 78.4 & 99.3 & 0.8 & 98.5 & 46.5 & 1.6 & 2.3 \\
MSB, PUA $\uparrow$             & 48.3 & 52.5 & 55.8 & 58.2 & 56.6 & 54.6 & 59.0 & \textbf{62.0} \\
MSB, NRP $\uparrow$             & 36.3 & 48.8 & 53.3 & 58.2 & 53.3 & 54.6 & 59.0 & \textbf{62.0} \\
\bottomrule
\end{tabular}}
\end{wraptable}
\subsection{Ablation Study}
\label{sec:ablations}

A complete $2\times2\times2$ factorial on 1167 frozen cases isolates path
confinement (P), capability and effect verification (C), and execution-boundary
adaptation (B). Eight arms share logical case identifiers, decoding, and judge,
with metric-specific eligibility: 9336 scored outcomes and no technical failure.
A0 is the wrapper-matched no-defense arm; A7 re-aggregates complete
\method{}-\textsc{p} on these cases. C jointly enables capability installation,
effect expansion, verification, and the combiner; its estimate measures that
bundle. The subset's task mixture differs from benchmarks.

In the ablation study, every C-containing arm closes PASB, unlike arms
without C. B alone leaves residuals on AgentDojo, AgentDyn, InjecAgent-Base,
and MCPTox, but every B-containing arm closes MSB-Call, against 53.8\% to
61.5\% without B. P alone reduces ASRs through suppression but leaves PASB and
MSB macro residuals. ASB and InjecAgent-Enhanced are zero across arms and
identify no layer. Several reduced arms tie A7's floor and none beats it in any
security row; the full combination has no demonstrated unique security
benefit. Adding B to C or P{+}C lowers MCPTox refusal at unchanged zero
valid-case ASR. C-containing arms partly recover P-only utility losses on AgentDojo,
AgentDyn, and PASB IPI. Paired A0-to-A7 AgentDyn utility falls
from 66.7\% to 33.3\%; the full-benchmark three-point bound does not hold on
this subset. ASB utility is implementation-sensitive
(\Cref{app:additional-results}). AgentDyn has 5842 nonempty candidate cuts
in 6647 guarded-call records but no executed nonterminal rewrite; utility
recovers through C restoring P's terminal blocks.

\section{Conclusion}

When admission evidence cannot distinguish behavior-changing specifications,
reviewing an agent artifact before execution cannot justify relaxing mediation
at the affected egress site. The calls it induces still require a runtime
decision. \method{} moves the
decision to the last enforceable tool boundary, where a concrete proposal,
authenticated authority, runtime provenance, and a working hook are all available
at once. Its certified contract combines represented-path separation, final-effect
authorization, and executor grounding. The evaluated configuration retains
effect authorization on restored calls, while path separation depends on the final intervention. On eight full benchmarks it lowers attack success while
losing at most three points of native utility relative to the undefended agent, and the paired ablation
shows the two checks are not interchangeable: effect verification carries most of
the security, boundary adaptation is what prevents indiscriminate refusal. The reduced-scale adaptive search obtains no success on 30
out-of-authority targets, within its 50-episode budget per case. Extending coverage to semantic restatement, fragmented flows, hidden
secondary effects, and durable distributed execution is future work.

\newpage

\subsection*{AI use statement}
We used generative AI tools to check grammar and spelling, revise wording and organization, and examine the manuscript's logical consistency. These tools also assisted in reviewing mathematical derivations, proof arguments, and theoretical analyses for internal consistency. No generative-AI system is a component of \method{}'s mediation logic; the evaluated agents and several baselines use models. We take responsibility for the final text, mathematical derivations, theoretical claims, experimental results, and artifacts.

\subsection*{Ethics statement}
This work studies attacks on deployed LLM agents and therefore carries dual-use
risk. The attacks use published benchmarks and search-generated variants, executed
only inside the corresponding benchmark sandboxes. The released artifact replaces harmful
payloads with inert local sinks and deterministic fixtures that preserve path
structure without providing working exfiltration targets
(\Cref{app:reproducibility}). No human subjects, no personal data, and no
production system were involved: every ``secret'' in our evaluation is a
synthetic credential minted for the run. The graph construction makes the
verifier's represented channels explicit; \Cref{sec:guarantee} identifies the
dependencies it does not cover. This distinction matters when interpreting the
security claims because a path outside the schema cannot enter a certified cut.
We report the coverage boundary alongside the formal guarantees, and the released
artifact uses inert local sinks and deterministic fixtures so readers can inspect
the graph construction without access to working exfiltration targets or
production systems.

\subsection*{Reproducibility statement}
The security model, the mediator, and the guarantee are specified so that a
reader can check them independently. \Cref{sec:execution-model,sec:security-property}
fix the execution model, the attacker, and the target property; \Cref{sec:method}
gives the four phases, the binding rule, the capacity function, and the decision
rule. The runtime results carry numbered assumptions: H1 to H5 open
\Cref{app:proofs}, and the protocol assumptions A1 to A9 sit in
\Cref{app:protocol}. The admission-time results of \Cref{sec:motivation} instead
carry named conditions: the $\delta$-achieving conservative transformation, C1,
and C2 in \Cref{sec:motivation}; almost-sure halting in \Cref{app:measure}; and
the hypotheses of \Cref{prop:code-resolved} with that proposition.
\Cref{prop:effect-sound} carries its own four assumptions in its statement.
\Cref{app:verify} lists the ten verifier conditions in the order the
implementation checks them, and \Cref{app:proofs} contains proofs of the formal claims
in the main text. \Cref{app:scope} states the premises no automatic check
reaches, and \Cref{app:protocol} marks the hardening protocol that lies outside
the evaluated artifact. For the empirical side, \Cref{app:hyperparameters} fixes
every configuration value before the primary run,
\Cref{app:benchmark-details,app:ablation-design} record scaffolds, evaluators,
denominators, and the ablation arms, \Cref{app:port-fidelity} records baseline
port fidelity, \Cref{app:statistics} fixes the statistical protocol and the
multiple-comparison policy, and \Cref{app:reproducibility} describes the
machine-readable manifest, the deterministic replay procedure, and the release
checklist. Anonymized source code
and the frozen manifest accompany the submission as supplementary material.



\bibliography{iclr2027_conference}
\bibliographystyle{iclr2027_conference}

\appendix
\newpage
\etocdepthtag.toc{mtappendix}

\vspace*{1mm}
\begin{center}
    \LARGE \bfseries Appendix of \method{}
\end{center}

\etocsettagdepth{mtchapter}{none}
\etocsettagdepth{mtappendix}{subsection}

{
  \hypersetup{
    colorlinks=true,
    linkcolor=kleinblue
  }

  \begingroup
  \makeatletter

  \let\oldl@section\l@section
  \renewcommand{\l@section}[2]{%
    \vspace{-5pt}%
    \oldl@section{#1}{#2}%
    \vspace{5pt}%
  }

  \let\oldl@subsection\l@subsection
  \renewcommand{\l@subsection}[2]{%
    \oldl@subsection{#1}{#2}%
    \vspace{5pt}%
  }

  \makeatother

  \tableofcontents
  \endgroup
}
\newpage

\numberwithin{theorem}{section}
\numberwithin{proposition}{section}
\numberwithin{lemma}{section}
\numberwithin{corollary}{section}
\numberwithin{definition}{section}
\numberwithin{remark}{section}

\section{Notations}

\begin{longtable}{@{}p{0.235\linewidth}p{0.725\linewidth}@{}}
\caption{List of notations and their descriptions used in the paper.}
\label{tab:notation}\\
\toprule
Notation & Description \\
\midrule
\endfirsthead
\multicolumn{2}{@{}l}{\emph{Table \thetable\ (continued)}}\\
\toprule
Notation & Description \\
\midrule
\endhead
\bottomrule
\endlastfoot

\multicolumn{2}{@{}l}{\emph{Agent execution and threat model (\Cref{sec:execution-model})}}\\[1mm]
$M$ & Frozen policy that alternates free text with tool calls \\
$\gT$ & Registered tool set \\
$a_t$ & Tool call proposed at step $t$ \\
$o_t$ & Observation returned by $a_t$ \\
$r$ & Answer returned to the user, itself an egress site \\
$\tau=\bigl((a_t,o_t)_t,r\bigr)$ & Trajectory of one episode \\
$\mathcal X$ & Standard Borel space of trajectories \\
$K=(\mathit{NL},\mathit{code})$ & Artifact: a skill, tool description, or MCP schema \\
$c\in\Ctx$ & Context: request, injected text, tool returns, state, secret \\
$s$ & Secret carried by the context \\
$\bar c\in\overline{\Ctx}$ & Public context with a designated secret slot \\
$\mu_{K,c}=\Traj(M,K,c)$ & Trajectory measure induced by $K$ under context $c$ \\
$\mu^{\gate}_{K,c}$ & Deployed trajectory measure, with $M$ run behind mediator $\gate$ \\
$\basefamily(K)=\{\mu_{K,c}\}$ & Behavior family: the security-relevant object, not $K$ itself \\
\midrule

\multicolumn{2}{@{}l}{\emph{Security property (\Cref{sec:security-property}, \Cref{app:measure})}}\\[1mm]
$\Egress(\tau)$ & Egress projection: ordered pairs $(\tmpl,\mathrm{bytes})$ leaving the deployment \\
$\Omega$ & Finite set of egress sites \\
$\mathcal B$, $\mathcal B^{*}$ & Standard Borel byte space; $\mathcal B^{*}=\bigsqcup_{n\geq0}\mathcal B^{n}$ \\
$\mathcal E$ & Codomain of $\Egress$, $\bigsqcup_{n\geq0}(\Omega\times\mathcal B)^{n}$ \\
$\tmpl$ & A single egress site \\
$\mathrm{pr}_{\tmpl}$ & Subsequence extractor selecting the entries tagged $\tmpl$ \\
$O_{\tmpl}=\mathrm{pr}_{\tmpl}\!\circ\!\Egress$ & Egress observed at site $\tmpl$, or $\bot$ if never reached \\
$\hat O_{\tmpl}$ & Value the code emits at $\tmpl$ before the mediator acts \\
$d(P,Q)$ & Total variation, $\sup_A|P(A)-Q(A)|$ over measurable $A$ \\
$\delta$ & Indistinguishability budget \\
$\Phi_{\mathrm{NI}}(K,\gate;\delta)$ & Global probabilistic-noninterference criterion (\Cref{eq:phi-ni}) \\
$\Phi_{\mathrm{site}}(K,\gate,\tmpl;\delta)$ & Per-site restriction of $\Phi_{\mathrm{NI}}$ (\Cref{eq:phi-site}) \\
$\Delta_{\gate}(K,\bar c,\tmpl)$ & Secret-induced distinguishability at $\tmpl$ under context $\bar c$ \\
\midrule

\multicolumn{2}{@{}l}{\emph{Gates and admission-time evidence (\Cref{sec:motivation})}}\\[1mm]
$\gate$ & Mediator: the deployed transformation applied at each egress site \\
$g$, $\gate_g$ & Gate mapping a judgment to $\{\mathsf{mediate},\mathsf{relax}\}$; the mediator it induces \\
$\gate_0$, $\gate_0^{\mathsf{r}}$ & Reference profile mediating every site but $\tmpl$; its relaxing extension \\
$\mathcal J_i(K,\tmpl)$ & Judgment held by admission route $i\in\{\mathrm{I},\mathrm{II},\mathrm{III}\}$ \\
$\gV(K,\tmpl)$ & Variant class: same code, interface, and site; different specifications \\
C1, C2 & Judgment-ambiguity and attacker-witness conditions of \Cref{thm:empty-unsound} \\
$\phi$ & Feature map through which a Route-III text extractor factors \\
$v$ & Slice inputs supplied to site $\tmpl$ along a trajectory \\
$F$, $F_0$ & Code-slice map $\hat O_{\tmpl}=F(v,s)$; its secret-independent common value \\
$\mathrm{med}$ & Mediator action at $\tmpl$, a measurable map of the call alone \\
\midrule

\multicolumn{2}{@{}l}{\emph{Evidence graph and phase 1 (\Cref{sec:graph}, \Cref{app:coverage})}}\\[1mm]
$H_t$ & Evidence graph: multigraph folded from the append-only log \\
$G^{-}_t$ & Active graph: committed edges at $\sigma_t$ plus the proposed edges of $a_t$ \\
$E^{\mathrm{fl}}$ & Flow subgraph; the only edge types reachability uses \\
E1--E5 & The five flow-edge emission rules \\
$q_t$ & Fresh, globally unique attempt identifier for $a_t$ \\
$\iota_t$ & External idempotency key, preserved across retries \\
$\mathrm{op}(q_t)$, $T(q_t)$ & Operation node and call-specific terminal of attempt $q_t$ \\
$\terminalbad$ & Set of terminals labeled unsafe by the versioned effect rule \\
$R_{v+1}$ & Resource version node minted by a confirmed declared write \\
$\nu$ & Provenance map over live values \\
$\sigma_t$ & Execution version: committed edges, labels, and live-value set \\
$\varepsilon_t$ & Policy epoch: schemas, effect rules, costs, contracts, adapter code \\
$L^{\mathrm{p}}_t$ & Canonical \texttt{PROPOSE} record of attempt $q_t$ \\
\midrule

\multicolumn{2}{@{}l}{\emph{Cut and certification, phase 2 (\Cref{sec:cut}, \Cref{app:verify})}}\\[1mm]
$P_t$ & Protected sources: data and resource nodes with protection labels \\
$\Psi$ & Finite family of deterministic re-encodings used by the derivation oracle \\
$\mathrm{match}(y,x)$, $\mathrm{canon}_{\psi}$ & Oracle match predicate (\Cref{eq:oracle}) and per-form canonicalization \\
$S_t$ & Untrusted-origin entries of the live-value set at $\sigma_t$ \\
$B_t$ & Untrusted frontier: closure of $S_t$ over committed flow edges \\
$U_t=\{T(q_t)\}\cap\terminalbad$ & Unsafe terminal of the pending call, empty for a safe call \\
$\mathrm{unsafe}_{\varepsilon_t}(a_t)$ & Registry predicate placing $T(q_t)$ in $\terminalbad$ \\
$\Impl(e)$, $\Sev(e)$ & Actions the adapter implements on $e$; the severing subset \\
$\mathrm{cost}(a,e)$ & Declared per-edge cost of action $a$, a strictly positive rational \\
$\bind_t(e)$ & Canonical binding: least-cost action in $\Sev(e)\cap\Impl(e)$, or $\bot$ \\
$\kappa(e)$ & Edge capacity (\Cref{eq:capacity}); finite only on bound proposed edges \\
$A$, $e_{\mathrm T}$ & Argument edges of $a_t$; its terminal edge $\mathrm{op}(q_t)\!\to\!T(q_t)$ \\
$R$, $A^{\star}$ & Nodes reachable from $B_t\cup P_t$ in the committed flow subgraph; $\{e\in A:\src(e)\in R\}$ \\
$\cut_t$ & Canonical minimum cut: cheapest, then lexicographically least \\
$\manifest_t$ & Enforcement manifest, one bound hook per cut-edge instance \\
$\mathrm{arg}(e)$, $\mathrm{hook}(e)$ & Registry-bound argument position and executor hook for edge $e$ \\
$\mathrm{infl}(P_t,a_t)$ & Influence test: a protected node reaches $\mathrm{op}(a_t)$, or the oracle matches \\
$\mathrm{rev}(a_t)$, $\mathrm{auth}(a_t)$ & Adapter reversibility flag; externally authenticated authorization flag; $\mathrm{low}(a_t)$ additionally requires empty sinks and no persistent effects \\
$\mathrm{dec}(a_t)$ & Decision rule (\Cref{eq:decision}) \\
R1--R6 & Refuse on undefined schema; R2 low-impact release; R3 authenticated, uninfluenced, C-certified release; refuse unauthorized original effects; refuse uncuttable or inadmissible manifests; refuse unauthorized final effects in the certified contract \\
$\mathrm{eff}(a_t)$ & Effect check: one when every effect atom of $a_t$ is certified \\
$\Cert_t$ & Path certificate $(q_t,L^{\mathrm{p}}_t,\sigma_t,\varepsilon_t,B_t\cup P_t,U_t,\cut_t,\manifest_t)$ \\
$n_t$, $n^{\ast}$, $h^{\mathrm{auth}}$ & Certificate sequence number; verifier counter and authorization head \\
$\mathsf{tok}_t$ & One-shot authorization token of the hardened protocol (\Cref{eq:token}) \\
$\mathsf{MAC}$, $\mathsf{h}$, $k$, $k'$ & MAC and hash; log key and token key, held by disjoint component pairs \\
\midrule

\multicolumn{2}{@{}l}{\emph{Enforcement protocol, phases 3--4 (\Cref{sec:enforce-phases}, \Cref{app:protocol})}}\\[1mm]
$z_j$ & Tag of log record $j$ in the append-only MAC chain (\Cref{eq:log}) \\
$\mathrm{res}\in\{\bot\}\cup\mathcal Q$ & Global single-flight reservation; $\mathcal Q$ is the space of attempt identifiers \\
$\mathrm{Apply}_{\varepsilon_t}$ & Deterministic canonical rewrite of the proposal under $\manifest_t$ \\
$E^{\mathrm{o}}(q_t,o_t)$ & Canonical output edge set bound into an outcome receipt \\
A1--A9 & Protocol assumptions: cryptography, atomicity, durability, state currency \\
\midrule

\multicolumn{2}{@{}l}{\emph{Causal chains and guarantees (\Cref{sec:guarantee}, \Cref{app:proofs})}}\\[1mm]
$\gamma=x_0\!\to\!\cdots\!\to\!\zeta_t$ & Real causal chain ending at the pending effect \\
$x_i$, $\zeta_t$ & Values along the chain; the pending unsafe effect \\
$w_i$, $w_i'$ & Calls realizing link $x_i\!\to\!x_{i+1}$ (writer and reader in the resource case) \\
$\alpha_{\mathrm D}$, $\alpha_{\mathrm T}$ & Simulation maps from chain values to data nodes and from the effect to a terminal \\
$\pi_i$, $\Pi$ & Emitted segment of link $i$; their concatenation, a flow path in $G^{-}_t$ \\
H1--H5 & Implementation conditions; H3 splits into H3a, H3b, H3c by label kind \\
\midrule

\multicolumn{2}{@{}l}{\emph{Configurations}}\\[1mm]
\method{}-\textsc{c} & Certified contract: no R2/R3 releases or restoration of blocks; the installed action preserves the cut and every dispatched final call passes R6 \\
\method{}-\textsc{p} & Evaluated configuration: evaluates all nine obligations; R3 requires all to pass, while R2 can release effects labelled reversible with empty sinks and no persistence despite failures. C can restore P's blocks, and declared repairs are dispatched without R6 \\
\end{longtable}

\section{Related Work}
\label{app:related-work}

\subsection{Agent Attack Surfaces}

Tool-integrated agents process external text in the same autoregressive context
that produces actions, so an observation can redirect execution.  InjecAgent
separates direct user harm from private-data exfiltration \citep{injecagent},
AgentDojo supplies executable user tasks with security targets and adaptive
interfaces \citep{agentdojo}, ASB broadens the surface to system, user, tool, and
memory stages \citep{asb}, and WASP separates partial compromise from terminal
attacker success in multi-step web tasks \citep{wasp}.  Those distinctions are
why we report native metrics instead of one pooled attack-success label.

Reusable artifacts differ from a single malicious page: the artifact persists,
can carry code, and defines a workflow.  MCPTox poisons the tool metadata of real
MCP servers and reports high attack success with near-zero refusal
\citep{mcptox}, MCPSecBench organizes the protocol's attack surface into a
reproducible benchmark \citep{mcpsecbench}, the MCP lifecycle survey describes
post-trust server changes
\citep{hou2025mcp}, while SkillAttack probes reusable skills through adversarial
prompts \citep{skillattack}. Trust inversion arises when an attacker-authored
description informs admission and shapes later execution. We therefore evaluate
prose, code-carried, and cross-turn payloads
separately.

Persistent state defeats a single-turn monitor, because source and sink can fall
in different turns.  AgentPoison optimizes triggers that retrieve poisoned
memories \citep{agentpoison}, ASB includes a memory-poisoning stage \citep{asb},
and PASB exercises prompt, tool, and memory-retrieval attacks in a personalized
local agent \citep{pasb}.  Our reduced-scale adaptive search tests out-of-authority AgentDojo goals
(\Cref{app:benchmark-details,app:additional-results}). Adaptive evaluation can
change conclusions drawn from native attacks \citep{zhan2025adaptive,nasr2026attacker};
out-of-band defenses also have a dedicated protocol and a small-scale adaptive
evaluation of Progent \citep{narisetty2026adaptive}.

\subsection{Agent Defense Paradigms}

CaMeL extracts control and data flow from a trusted request and wraps a system
layer around a susceptible model \citep{camel}.  FIDES formalizes the properties
dynamic taint tracking can enforce and builds an IFC planner with confidentiality
and integrity labels \citep{fides}.  RTBAS adapts IFC to tool-based agents and
falls back on user confirmation when automated screening cannot establish
integrity or confidentiality \citep{rtbas}. CaMeL derives its plan from a trusted
request, while FIDES uses trusted wrappers and flow labels. The deployments we
target also expose attacker-authored artifact descriptions to the planner. Our
comparison uses instrumented IFC implementations where their required hooks
are available, preserving the scope of their flow-control mechanisms.

Progent expresses least privilege as symbolic rules over tool names and arguments
\citep{progent}, Task Shield checks whether instructions and calls serve the
user's stated task \citep{taskshield}, and ClawGuard enforces user-confirmed
rules at tool boundaries across web, local, MCP, and skill channels
\citep{clawguard}.  IPIGuard uses a tool dependency graph \citep{ipiguard}, Conseca derives
contextual policies \citep{conseca}, and FORGE enforces formal policies over
tool operations \citep{forge2026}. These approaches share \method's concern
with action authorization and differ in the evidence used.  Their rules range over names, arguments, and
task descriptions; \method{} derives a separating cut from an execution graph and
then demands a concrete hook binding for every cut edge.

MELON detects indirect injection by comparing the original action with a
masked-prompt re-execution \citep{melon}, DataSentinel trains an adaptive
detector through a minimax objective \citep{datasentinel}, and LlamaFirewall
combines prompt, agent-alignment, and code guardrails \citep{llamafirewall}.
Structured-query and model-level defenses separate instructions from data or
train the model to respect the distinction \citep{struq,secalign}.  All of these
compose with \method{}, and our evaluation keeps a pre-filter's AUROC, TPR, and
FPR apart from end-to-end agent security.

AgentArmor reconstructs runtime traces as program graphs and applies a type
system to security metadata \citep{agentarmor}. More recent monitors operate at
argument and execution granularity: PACT supports fine-grained provenance and
scoped discharge certificates \citep{pact2026}, AuthGraph separates observed
provenance from request-derived authorization \citep{authgraph2026}, and CXI
binds executed effects to a canonical manifest and requires fresh authority
after relevant edits or repairs \citep{cxi2026}. Thus fine-grained authorization
and re-authorization alone do not distinguish \method{}. Its analysis couples a
candidate graph cut to the intervention actually installed: separation survives
exactly when that intervention blocks the terminal or severs every reachable
argument edge, alongside authorization of the final call
(\Cref{sec:guarantee,app:proof-cut}). The evaluated restoration and repair paths
make this distinction observable in the component study. IsolateGPT establishes
isolation between LLM applications and the host \citep{isolategpt}, a separate
boundary from these call-level checks.

\subsection{Formal Foundations}

Complete mediation requires every security-relevant access to pass through a
tamper-resistant, analyzable mechanism \citep{saltzer1975,anderson1972}, and
runtime policy enforcement characterizes which trace properties execution
monitors can enforce \citep{schneider2000}.  \method{} is an edit automaton rather than a truncation automaton, since its
actions rewrite and suppress arguments as well as halting the call, and that
class is strictly more expressive \citep{ligatti2005,hamlen2006}.  Lattice-based confidentiality
\citep{denning1976} and language-level information-flow control
\citep{myers1999jif,sabelfeld2003} give the basis for source and sink labels, and
TaintDroid and libdft show that labels do propagate through instrumented
execution \citep{taintdroid,libdft}.  Our opaque-boundary ablation is built to
respect that capability, not to deny it.

Proof-carrying code attaches machine-checkable evidence that code satisfies a
policy \citep{necula1997}, and model-carrying code checks a model of an untrusted
program against a safety property \citep{sekar2003}.  Both apply when the code slice is the whole security-relevant semantics
\emph{and} the model's choice of that slice's inputs is itself
secret-independent, which is the positive case of \Cref{prop:code-resolved}.  \method's runtime guarantee is deliberately weaker
and more operational.  The max-flow/min-cut theorem gives exact
minimum-capacity separation for a fixed capacitated graph \citep{fordfulkerson},
and our solver applies it directly; the registry costs are deterministic inputs,
not calibrated probabilities.

\section{Protocol and Verification}

\subsection{Measure-Theoretic Setup}
\label{app:measure}

One episode gives a trajectory
$\tau=\bigl((a_t,o_t)_t,r\bigr)$: the calls $a_t$ proposed against the registered
tool set $\gT$, the observations $o_t$ they return, and the answer $r$, which is
itself an egress site. Trajectories live in a standard Borel space $\mathcal T$
and episodes halt almost surely, so $\Egress$ is defined off a null set. Its codomain is
$\mathcal E=\bigsqcup_{n\geq0}(\Omega\times\mathcal B)^n$, for the finite site set
$\Omega$ and a standard Borel byte space $\mathcal B$, carrying the
disjoint-union $\sigma$-algebra; $\Egress$ is measurable for it. The per-site projection of \Cref{sec:motivation} maps into
$\mathcal B^{*}\sqcup\{\bot\}$, where $\mathcal B^{*}=\bigsqcup_{n\geq0}
\mathcal B^{n}$ is again standard Borel, and \Cref{lem:site-necessary} proves it
measurable. The artifact $K$ paired with a
context $c$ induces $\mu_{K,c}=\Traj(M,K,c)$ on $\mathcal T$, so the dependence on
the frozen policy and the artifact is explicit; running $M$ behind a mediator
$\gate$ gives $\mu^{\gate}_{K,c}$.

\subsection{Graph Fold and Emission Rules}
\label{app:coverage}

Parallel edge instances stay distinct in the multigraph.  Besides the nodes of
\Cref{def:graph} the fold also records the agent's stated intentions and plans,
and the control and intent edges among them stay in the audit record and never
reach a cut.  Values and effects connect
through an operation node, so the construction has no direct
\texttt{Data}-to-\texttt{Data} and no \texttt{Data}-to-\texttt{Terminal} edge;
\Cref{app:proof-transparent} uses both facts.

The host maintains a provenance map $\nu$ for live values and applies five rules,
which are the complete set of flow-edge emitters.  (E1) Before certification,
each schema-defined proposal argument emits an edge from its data or resource
node to the fresh operation node.  (E2) A valid outcome receipt emits an edge from
the operation to every schema-defined returned value, minting a fresh provenance
identifier unless the adapter certifies the value as an unchanged pass-through.
(E3) Such a pass-through keeps its identifier at the next declared position.
(E4) Every proposal creates a fresh terminal $T(q_t)$ and emits
$\mathrm{op}(q_t)\to T(q_t)$ before dispatch, and the versioned unsafe-effect
rule labels that terminal.  (E5) A confirmed outcome of a schema-declared write
mints a new resource version node $R_{v+1}$, commits
$\mathrm{op}(q)\to R_{v+1}$, and gives $R_{v+1}$ the union of the origin and
protection labels of that call's committed inputs.  Rules E2 and E5 fire at
finalization, so their edges are committed and not proposed.  These rules emit
structure, not a claim about the semantics of hidden computation.

Instance status follows from the log: a refusal expires the current proposal
edges, \texttt{OUTCOME} commits them together with the authenticated output
edges, \texttt{UNKNOWN} commits the proposal alone, and \texttt{BLOCKED} or
\texttt{CANCEL} commits no flow.  The frontier $B_t$ is folded as in
\Cref{sec:cut} from the origin-labeled live-value set, over committed flow edges
only, so $B_t\cap\terminalbad=\emptyset$ and no node enters $B_t$ through a
current proposed edge; in particular the fresh operation and terminal of $q_t$
stay out.  The disjointness $P_t\cap U_t=\emptyset$ recorded in \Cref{sec:cut}
follows for the same reason, since protection labels apply only to data and
resource nodes.

\begin{lemma}[Every unsafe path crosses a current edge]
\label{lem:cuttable}
Let $E^{\mathrm p}(q_t)$ be the proposed flow-edge instances of $a_t$. Every
directed path from $B_t\cup P_t$ to nonempty $U_t$ contains an edge in
$E^{\mathrm p}(q_t)$. If the terminal in-edge admits \texttt{block\_call}, that
edge alone is a finite-capacity cut.
\end{lemma}

\begin{proof}
Nonempty $U_t$ contains only $T(q_t)$, which is fresh, disjoint from the source
set, and by (E4) has the single proposed in-edge $\mathrm{op}(q_t)\to T(q_t)$.
Every path to it therefore ends with that current edge, and removing it separates
the terminal. When it has a nonempty binding to \texttt{block\_call},
\Cref{eq:capacity} gives it a finite value.
\end{proof}

Declared costs are strictly positive rationals with a fixed denominator, which
the implementation scales to integers.  Positivity matters: with a zero-cost
action a minimum-cost separating set could be padded arbitrarily, and the
verifier's set equality in \Cref{app:verify} would be ill posed.  The infinity sentinel exceeds the sum of every finite capacity, and acceptance
separately rejects a cut containing it.  Nodes and edge instances carry a
canonical order with the terminal edge first, which is what makes the selection in
\Cref{prop:cut-closed-form} single-valued.  Builder and verifier therefore agree
without floating-point arithmetic and without depending on which minimum cut a
flow solver happens to return.

\subsection{Deployment Hardening: Authenticated Execution Protocol}
\label{app:protocol}

The evaluated prototype mediates one in-process execution path and invalidates a
decision when the episode state it names changes.  This subsection specifies the
stronger protocol a distributed deployment needs, with durable records, a global
reservation, and one-shot token spending.  It is what discharges the freshness
condition H5 of \Cref{app:proofs} outside a single process, and it is not part of
the artifact we evaluate in \Cref{sec:experiments}.

Every record belongs to one append-only MAC chain in the style of tamper-evident
audit logs.  With a canonical, length-delimited encoding,
\begin{equation}
  z_j=\mathsf{MAC}_k\bigl(z_{j-1}\|\mathrm{enc}(\mathrm{type}_j,q_j,
  \mathrm{payload}_j)\bigr).
  \label{eq:log}
\end{equation}
A record contains its type, attempt identifier, payload, previous tag, and tag.
Types are \texttt{PROPOSE}, \texttt{DECIDE}, \texttt{DISPATCH}, \texttt{OUTCOME},
\texttt{BLOCKED}, \texttt{CANCEL}, \texttt{UNKNOWN}, and \texttt{RECONCILE}.
$\textsc{AcceptIfFresh}$ is the primitive that checks \Cref{app:verify}, moves
the reservation from $\bot$ to $q_t$, advances the verifier sequence and
authorization head, appends \texttt{DECIDE(accept)}, and returns
\begin{equation}
  \mathsf{tok}_t=\mathsf{MAC}_{k'}\bigl(q_t\|n_t\|\sigma_t\|\varepsilon_t\|
  \mathsf{h}(L^{\mathrm{p}}_t)\|\mathsf{h}(\manifest_t)\bigr),
  \label{eq:token}
\end{equation}
where $n_t$ is the verifier's next sequence number and $k'$ is shared only by the
verifier and the executor.  The log key $k$ and the token key $k'$ sit with
different pairs of components, so neither the graph builder nor the certificate
producer holds either one.  The executor spends $\mathsf{tok}_t$ once by
compare-and-swap, rechecks the versions, and holds each resource lease through
the access.  Its permissive counterpart $\textsc{ReserveIfFresh}$ checks the
same state and epoch for R2 and R3, takes the same reservation, and appends
\texttt{DECIDE(release)} without a path certificate or token; R3 still
requires C's nine obligations to pass.  A stale proposal is
refused, never silently rebuilt.  For a nonblocking manifest,
$\mathrm{Apply}_{\varepsilon_t}$ is the deterministic canonical adapter function
that applies every certified rewrite to the call bytes in $L^{\mathrm p}_t$; it
preserves $\iota_t$ and every versioned resource identity, and it is undefined for
a terminal block or an inadmissible manifest.

The remaining protocol content is the assumption list the proofs cite.  Nine
assumptions carry the protocol obligations of H5 in \Cref{app:proofs}, together
with the currency obligation A9 supplies.

\emph{Cryptography.} (A1) $\mathsf{MAC}$ is existentially unforgeable under
chosen-message attack, with $k$ held by the host and log store and $k'$ only by
the verifier and the executor, so neither the graph builder nor the certificate
producer holds either key.  (A2) $\mathsf{h}$ is collision resistant.  (A3)
Calls, manifests, and log payloads use a canonical length-delimited
serialization, so the concatenations in \Cref{eq:log,eq:token} parse uniquely.

\emph{Identity and atomicity.} (A4) Attempt identifiers $q_t$ are globally unique
and never reused, while $\iota_t$ is a separate field that may repeat; the
\texttt{PROPOSE} record stores the immutable call bytes, all resource versions,
and $\iota_t$, so $\iota_t$ is covered by $\mathsf{h}(L^{\mathrm p}_t)$ and cannot
change after certification.  (A5) Log append, $\textsc{AcceptIfFresh}$,
$\textsc{ReserveIfFresh}$, $\textsc{Finalize}$, $\textsc{Reconcile}$, epoch
updates, and the executor's compare-and-swap and $\textsc{LeaseIfCurrent}$ are
linearizable, and the appender and the four mediation primitives share one
critical section.  A successful reservation atomically appends its decision, and a
failure records refusal only while the attempt is undecided.  (A6) Five objects
are distinct and rollback-protected: the raw log tip, which moves on every append
and is therefore never the freshness anchor; the execution version $\sigma_t$,
naming committed edges, execution-derived labels, and the live-value set; the
verifier store $(n^{\ast},h^{\mathrm{auth}})$, which advances only on certificate
acceptance; the policy epoch $\varepsilon_t$, naming an immutable snapshot of
labels, schemas, effect rules, costs, contracts, and adapter code; and a live or
halted mode together with one global reservation
$\mathrm{res}\in\{\bot\}\cup\mathcal Q$, where $\mathcal Q$ is the space of
attempt identifiers.

\emph{Durability and executor.} (A7) Log appends are durable and write-ahead, so
\texttt{DISPATCH} commits to stable storage before any externally visible send of
that call and no send is reordered ahead of that commit.  A trusted adapter
returns \texttt{BLOCKED} only for a terminal \texttt{block\_call} after enforcing
non-dispatch, and an outcome receipt only when bound to
$(q_t,\iota_t,o_t,\mathsf{h}(E^{\mathrm{o}}(q_t,o_t)))$, with $E^{\mathrm{o}}$ the
canonical output edge set.  (A8) $\textsc{Prepare}$ performs the dispatch sequence of \Cref{sec:enforce-phases}
under one further discipline: it acquires the policy and resource leases and spends
the token only after checking $q_t$, $\sigma_t$, $\varepsilon_t$,
$\mathsf{h}(L^{\mathrm p}_t)$, and $\mathsf{h}(\manifest_t)$ against the current
stores and the immutable \texttt{PROPOSE} record.  It sends nothing before that
check, rolls back partial installation, and returns only \texttt{CANCEL}, a
$(q_t,\mathsf{h}(\manifest_t),\texttt{BLOCKED})$ receipt, or a
$(q_t,\mathsf{h}(\manifest_t),\mathsf{h}(\mathrm{Apply}_{\varepsilon_t}
(\manifest_t,L^{\mathrm p}_t)),\texttt{READY})$ receipt whose branch dispatches
only those bytes and holds every resource lease through the certified access.
(A9) Every event that could add a committed flow edge, change a label, or change
the live-value set either passes the mediator or advances $\sigma_t$ before the
next certificate is issued.  A resource adapter must supply a snapshot or a lease
held through the certified access, or the position is schema-undefined.  A9 is
what makes $\sigma_t$ a statement about the world and not only about the store,
and an out-of-band write through a channel the deployment never registered
violates it.

Finalization follows \Cref{sec:enforce-phases} and is atomic, with reconciliation
opening a new attempt $q_{t'}$ that carries the same $\iota_t$.  Crash recovery
uses the same durable intent: without
\texttt{DISPATCH}, write-ahead ordering proves no send occurred and the attempt
can be cancelled, while a dispatch without a receipt becomes \texttt{UNKNOWN}.
Rule (E5) fires only on a confirmed outcome, so an \texttt{UNKNOWN} write mints no
version node, and reconciliation must re-read the resource before any later step
may use it.  The protocol therefore gives one-shot token use, not exactly-once
delivery without adapter support.

\subsection{The Recomputing Verifier}
\label{app:verify}

Inside $\textsc{AcceptIfFresh}$, the verifier folds the authenticated log and the
current registry snapshot. It reconstructs $G^{-}_t$, $B_t$, $P_t$, $U_t$,
flow-edge instances, bindings, capacities, and statuses, and reads no graph,
capacity, or verdict from the certificate producer. It accepts $\Cert_t$ only if
every check below passes.

\begin{enumerate}[leftmargin=1.4em,itemsep=0.1em,topsep=2pt]
\item The mode is live and $\mathrm{res}=\bot$.
\item $n_t$ is the next verifier sequence, $\sigma_t$ is current, and
$L^{\mathrm p}_t$ is the complete immutable \texttt{PROPOSE} record of $q_t$
with no decision for that attempt.
\item $\varepsilon_t$ is the current policy epoch.
\item The operation and every argument are schema-defined. Immutable proposal
bytes and resource versions agree with $L^{\mathrm p}_t$. The verifier rebuilds
$B_t$ from the origin-labeled live-value set stored with $\sigma_t$, and rebuilds
$P_t$ from committed protection labels together with a rerun of
\Cref{eq:oracle} on the frozen proposal.
\item The verifier reruns the unsafe-effect predicate to construct $U_t$, and
every supplied cut edge is a proposed flow-edge instance of $q_t$.
\item For every proposed flow edge, it recomputes the canonical binding in
\Cref{eq:binding} and the capacity in \Cref{eq:capacity}.
\item Removing $\cut_t$ separates $U_t$ from $B_t\cup P_t$, its capacity is
finite and equals the declared cost, and an empty target carries the empty cut.
The verifier checks the cut it is given rather than re-solving the flow problem.
\Cref{prop:cut-closed-form} shows that an independent recomputation would cost
one reachability pass, which we note as an available strengthening rather than an
implemented one.
\item $\mathrm{dom}(\manifest_t)=\cut_t$, with exactly one canonically encoded
entry per cut-edge instance and no repeated edge.
\item Every cut edge has a nonempty binding, and every manifest entry contains
that binding with the registry's argument position and hook.
\item The adapter declares the complete manifest jointly admissible: rewrites
do not conflict, every hook is timely, and each severance contract is declared
preserved under the complete set.
\end{enumerate}

Condition 6 blocks substitution of a costlier action, condition 7 blocks a cut
that does not separate or that misstates its cost, and conditions 8 and 9 block
an apparently separating cut with missing bindings.  Condition 10 checks a
registry declaration and not its semantic truth, an obligation H4 carries at
execution level.  Failure returns $\bot$ without touching the reservation or
the verifier stores, except that an undecided attempt receives one refusal record.

\subsection{Details Deferred from the Main Text}
\label{app:deferred}

\tightparagraph{The untrusted frontier.}
Let $S_t$ hold the untrusted-origin entries of the live-value set at $\sigma_t$.
The verifier computes $B_t$ by closing $S_t$ over committed flow edges and
dropping terminals, which loses nothing since a terminal has no out-edge.
Proposed edges never advance the closure. If they did, one untrusted argument
would collapse every call into terminal blocking, since the proposal itself would
carry the frontier to the operation node it just created.

\tightparagraph{Severance contracts.}
An action lies in $\Sev(e)$ when its warrants guarantee that, once installed, no
execution of this call takes a real link whose represented segment contains $e$.
The contract is about the call and not about the graph, so observation and user
confirmation sever nothing: they change what the operator sees, not what the call
can do. This is also why $\Sev$ is declared by the deployment while $\Impl$ is
declared by the adapter. The adapter knows which hooks exist; only the deployment
can assert that installing one removes a real dependency.

\tightparagraph{The two redaction warrants.}
Redaction on a data edge lies in $\Sev(e)$ only under two declarations. The
adapter warrants that the emitted value is independent of the original in
content, length, and structure, so a length-preserving mask does not qualify. The
registry warrants, per operation and argument position, that the call's
occurrence does not depend on that argument. Its default is that occurrence may
depend, since no adapter knows whether the model would have proposed this call
with a different argument. Miss either warrant and the argument edge has no
severing action. The terminal edge then admits \texttt{block\_call}, and a
read-only resource view only when the registry certifies that the call's single
unsafe effect is the write. This default is why \method{} blocks at most sites
rather than editing them.

\subsection{Scope and Residual Assumptions}
\label{app:scope}

Four cases carry no structural link between a real dependence and the graph:
undeclared reads, secondary external calls made inside a tool, strings recreated
by the model, and external changes to resources that bypass versioned adapters.
R1 refuses an undefined operation or argument but does not detect a hidden
read inside an otherwise schema-defined tool. Out-of-band resource changes
violate A9.  The transform family of \Cref{eq:oracle} is
$\Psi=\{$raw text, case normalization, whitespace normalization, base64,
base64url, hex, percent encoding, the common digests, fixed-length prefixes$\}$,
which narrows the third case to deterministic re-encodings; paraphrase,
translation, arithmetic encodings, and short fragments still escape.  It also errs in the other
direction: short low-entropy values match
benign text, so $P_t$ grows spuriously and R3 inherits both behaviors, which is
one reason that rule exists only in \method{}-\textsc{p}.

The cost in \Cref{eq:capacity} is a declared per-edge surrogate and not episode
utility, so one call-level block can cover several edges while the sum charges
each.  Minimality is therefore over per-edge-groundable sets rather than over
jointly admissible ones, and a cheaper conflicting cut may coexist with a costlier
compatible one, where the certified contract fails closed. Canonical binding
followed by R6 is not a complete search over authorized interventions: a
cheapest rewrite can fail a final argument constraint while another declared,
severing rewrite satisfies it. The cut's cost optimality does not resolve that
action-selection problem.

Request compilation, effect labels, validators, source labels, and severance
contracts are semantic premises; runtime invocation counts do not establish them.  The verifier shares the
graph-fold code, the label rules, $\terminalbad$, and the registry with the
builder, so the two agree on every semantic omission and neither tests H1 to H4.
A1 to A3 must hold, or the token of \Cref{eq:token} records its fields instead of
binding them.  Every severance contract must be true and not only declared.  And
nothing here defends against a compromised mediation layer: a party holding the
log key $k$ can fabricate the observation prefix that every later decision reads.
Reversibility alone says nothing about disclosure. R2 additionally requires
empty sinks and no persistent side effects. Semantic soundness here means that
every irreversible effect is labelled irreversible and every outward effect is
either labelled irreversible or declares a nonempty sink or persistence. Under
that premise R2 releases no outward effect, but it still supplies no path
certificate and can bypass authorization or budget failure. For fixed trusted
capabilities, schema, provenance, and event state, the checks are deterministic.
An injection can change subsequent proposals and provenance, but cannot directly
mint or widen authority from the trusted request. Residual goals can fall inside
the compiled authority, use declared repairs or the R2 class, exploit
unrepresented channels, or violate these semantic premises.

\section{Proofs}
\label{app:proofs}

The runtime proofs use five implementation conditions, H1 to H5, with the
labeling condition H3 split by label kind.  In a real causal chain, a
value link means the later value semantically depends on the earlier one, and the
final link means the effect's occurrence or content depends on the last value;
both are defined without reference to the graph.  H1: the chain is emitted and
simulated in the sense of \Cref{sec:guarantee}, so every covered link emits its
directed segment before dispatch, the segments concatenate, and taking a link
entails traversing its segment.  H2: every relevant unsafe or state-changing
effect passes the mediator.  H3a: protected-source labels are complete for the
chain under consideration; H3b: untrusted-origin labels are; and H3c: the
unsafe-effect rule labels the terminal of the pending call, so $U_t\neq\emptyset$.
H4: the
severance contract of every installed manifest entry is true under the complete
manifest, for the adapter content warrant and the registry occurrence warrant
alike.  H5: the decision is
current at dispatch, so the proposal, the execution version, and the policy epoch
are rechecked, the only dispatched call is the canonical rewrite of the frozen
proposal, and each resource snapshot or version is held through the certified
access.  \Cref{def:grounded} is H4 together with the single-dispatch clause of
H5.  The prototype checks freshness process-locally; restored and repaired actions
must separately satisfy the final-action premises to instantiate H4 and H5.
Assumptions A1 to A9 of
\Cref{app:protocol} discharge it in a distributed deployment, and
\Cref{prop:recertification,prop:view-freshness} give the argument.

\subsection{Admission-Time Results}
\label{app:proof-theorem1}

\tightparagraph{The gate model.}
A gate $g$ maps a judgment $\mathcal{J}_i(K,\tmpl)$ to
$\{\mathsf{mediate},\mathsf{relax}\}$ and induces a mediator $\gate_g$ that
applies its \emph{conservative transformation} at every site it mediates and the
identity elsewhere. That transformation is assumed $\delta$-achieving, so
\Cref{eq:phi-site} holds at every mediated site for every artifact in the gate's
declared domain; refusing the call is the canonical instance, since it makes
$O_{\tmpl}$ constant. Both sides of \Cref{eq:phi-site} use the deployed measure.
We write $\gate_0$ for a reference profile that mediates every site but $\tmpl$,
and $\gate_0^{\mathsf{r}}$ for its extension that relaxes $\tmpl$. The variant
class $\gV(K,\tmpl)$ collects the artifacts in the domain sharing code,
interface, and egress site while carrying different attacker-written
specifications.

\tightparagraph{The three routes.}
$\mathcal{J}_{\mathrm{I}}$ holds facts discharged by verification conditions over
the code. It proves that the mail implementation never reads credentials and says
nothing about how $\mathit{body}$ was chosen. $\mathcal{J}_{\mathrm{II}}$ is a
sound abstraction, either of the artifact against its interface or of agent
behavior reconstructed from traces; the code-interface form inherits the same
blind spot, while an abstraction of full agent behavior can succeed, but only by
separating the benign completion from the malicious one.
$\mathcal{J}_{\mathrm{III}}$ reads dependencies off the specification and fails
whenever a semantics-changing rewrite leaves the extractor's output unchanged.
Publisher attestation replaces such reasoning with identity and stays out of
scope.

\tightparagraph{The two conditions.}
Write $\Delta_{\gate}(K,\bar c,\tmpl)$ for the inner supremum of
\Cref{eq:phi-site}, the secret-induced gap at $\tmpl$ under context $\bar c$.
C1: route $i$ assigns one judgment to a pair $K_0,K_1\in\gV(K,\tmpl)$ in the
gate's domain with
$\sup_{\bar c}\Delta_{\gate_0^{\mathsf{r}}}(K_0,\bar c,\tmpl)=0$, so $K_0$ is
safe to relax. C2: the attacker can pick $\bar c^\star$ with
$\Delta_{\gate_0^{\mathsf{r}}}(K_1,\bar c^\star,\tmpl)>\delta$, so $K_1$ is not.
They are what a route and a threat model must supply before
\Cref{thm:empty-unsound} says anything.

\begin{lemma}[The site criterion is necessary]
\label{lem:site-necessary}
Work in the setting of \Cref{app:measure}, write
$\mathcal B^{*}=\bigsqcup_{n\geq0}\mathcal B^{n}$, and let
$\mathrm{pr}_{\tmpl}:\mathcal E\to\mathcal B^{*}\sqcup\{\bot\}$ send a finite
egress sequence to the subsequence of its entries tagged $\tmpl$, or to $\bot$
when there is none, so that $O_{\tmpl}=\mathrm{pr}_{\tmpl}\circ\Egress$. Then $\mathrm{pr}_{\tmpl}$ is
measurable, and for every mediator $\gate$, artifact $K$, public context $\bar c$,
and pair $s_0,s_1$,
\[
  d\bigl((O_{\tmpl})_{\#}\mu^{\gate}_{K,\bar c[s\leftarrow s_0]},\,
         (O_{\tmpl})_{\#}\mu^{\gate}_{K,\bar c[s\leftarrow s_1]}\bigr)
  \;\leq\;
  d\bigl(\Egress_{\#}\mu^{\gate}_{K,\bar c[s\leftarrow s_0]},\,
         \Egress_{\#}\mu^{\gate}_{K,\bar c[s\leftarrow s_1]}\bigr).
\]
Consequently $\Phi_{\mathrm{NI}}(K,\gate;\delta)$ implies
$\Phi_{\mathrm{site}}(K,\gate,\tmpl;\delta)$ at every egress site.
\end{lemma}

\begin{proof}
On each summand $(\Omega\times\mathcal B)^n$ of $\mathcal E$ the map
$\mathrm{pr}_{\tmpl}$ is a finite composition of coordinate projections with
selection on the value of the discrete tag, hence measurable; a map on a countable
disjoint union that is measurable on each summand is measurable, so
$\mathrm{pr}_{\tmpl}$ and therefore $O_{\tmpl}$ are measurable. Total variation
does not increase under a measurable map: for any measurable $f$ and any $P,Q$,
$d(f_{\#}P,f_{\#}Q)=\sup_{A}|P(f^{-1}A)-Q(f^{-1}A)|\leq\sup_{B}|P(B)-Q(B)|=d(P,Q)$,
since the preimages $f^{-1}A$ form a sub-collection of the measurable sets. Apply
this with $f=\mathrm{pr}_{\tmpl}$ to the two egress pushforwards. Taking the
supremum over $s_0,s_1$ bounds $\Delta_{\gate}(K,\bar c,\tmpl)$, and taking the
supremum over $\bar c$ gives the second claim. The converse fails and we never
use it.
\end{proof}

\begin{proof}[Proof of \Cref{thm:empty-unsound}]
Let $K_0,K_1\in\gV(K,\tmpl)$ witness C1 and let $g$ be as in the statement. Both
variants lie in $g$'s declared input domain and carry the same judgment, and $g$
has the form $g(\mathcal J_i(\cdot,\tmpl))$, so it decides them alike. Suppose
that decision is $\mathsf{relax}$. Because $\gate_g$ agrees with $\gate_0$ off
$\tmpl$ at every artifact in the domain, in particular at $K_1$, we have
$\gate_g=\gate_0^{\mathsf{r}}$ there, and C2 supplies a public context
$\bar c^\star$ with $\Delta_{\gate_g}(K_1,\bar c^\star,\tmpl)>\delta$. Hence
$\Phi_{\mathrm{site}}(K_1,\gate_g,\tmpl;\delta)$ fails at an artifact inside the
gate's domain, contradicting $\delta$-soundness. Every such gate therefore
mediates both variants, and by \Cref{lem:site-necessary} the same failure
forecloses $\Phi_{\mathrm{NI}}$ at level $\delta$. The statement has content only
because the class of $\delta$-sound gates is nonempty: the gate that mediates
everywhere is $\delta$-sound by the $\delta$-achieving assumption, not vacuously.

It remains to say when C1 holds. For Route I and a code-interface Route II
abstraction the judgment is a function of the shared code and interface alone, so
it is constant on $\gV(K,\tmpl)$. C1 then reduces to the existence inside
$\gV(K,\tmpl)$ of a safe member $K_0$ and, with C2, of a leaking member $K_1$.
That is a $\delta$-separation condition on the induced families
$\basefamily(K_0)$ and $\basefamily(K_1)$, hence a property of the frozen policy
$M$ and the context space rather than of the code. It fails whenever the
hypothesis of \Cref{prop:code-resolved} holds, and it also fails when no safe
member exists at all, so the two results do not exhaust the space between them. The argument does not
cover a sound abstraction of full agent behavior that separates the pair. For
Route III, let the extractor factor through a feature map $\phi$. If a
semantics-changing rewrite preserves $\phi$, preserves the code and interface, and carries one specification to
$\sup_{\bar c}\Delta_{\gate_0^{\mathsf r}}(\cdot,\bar c,\tmpl)=0$ while the other
satisfies C2, then the two carry the same extracted judgment and C1 holds. This is a sufficient condition for a blind spot,
not a claim about every text extractor. Identity whitelists backed by trusted
attestation, and independently certified invariants, fall outside the route
judgment considered here.
\end{proof}

\begin{proposition}[The code-resolved positive case]
\label{prop:code-resolved}
Fix $\tmpl$ and the deployed mediator $\gate$ with that site relaxed and
its policy at other sites held fixed. Let $v(\tau)$ collect the slice inputs
supplied to $\tmpl$ along $\tau$, a measurable function of the trajectory.
Write $\hat O_{\tmpl}$ for the value the code emits at $\tmpl$ before the mediator
acts on it. Suppose that for every $L\in\gV(K,\tmpl)$ there is a measurable $F$
with $\hat O_{\tmpl}(\tau)=F\bigl(v(\tau),s\bigr)$ for
$\mu^{\gate}_{L,\bar c}$-almost every $\tau$ and every $\bar c$, so that $v$ also
determines how often $\tmpl$ is reached. Suppose a sound verification condition
establishes $F(v,s_0)=F(v,s_1)$ for all $v,s_0,s_1$, including the site's
occurrence, termination, and exposed metadata, and write $F_0$ for the common
value, a measurable function of $v$ alone. Suppose $\gate$ acts at $\tmpl$ as a measurable function of the $\tmpl$-call
alone, and suppose finally that under the deployed measure the law of $v$ is
independent of the secret,
$v_{\#}\mu^{\gate}_{L,\bar c[s\leftarrow s_0]}=v_{\#}\mu^{\gate}_{L,\bar c[s\leftarrow s_1]}$
for every $\bar c$ and $s_0,s_1$. Then $\Delta_{\gate}(L,\bar c,\tmpl)=0\leq\delta$,
so a Route I gate may relax the class without violating $\Phi_{\mathrm{site}}$.
\end{proposition}

\begin{proof}
Fix $L$, $\bar c$, and $s_0,s_1$. The verification condition gives
$\hat O_{\tmpl}=F_0(v)$ almost surely, and $F_0$ does not read $s$. The
mediator's action at $\tmpl$ is a fixed measurable map $\mathrm{med}$ of the
call, applied at each occurrence, so $O_{\tmpl}=G(v)$ for the fixed measurable
$G=\mathrm{med}\circ F_0$. The last hypothesis makes the law of $v$ the same
under $s_0$ and $s_1$. Equal laws pushed
through one measurable map are equal, so
$(O_{\tmpl})_{\#}\mu^{\gate}_{L,\bar c[s\leftarrow s_0]}
=G_{\#}v_{\#}\mu^{\gate}_{L,\bar c[s\leftarrow s_0]}
=G_{\#}v_{\#}\mu^{\gate}_{L,\bar c[s\leftarrow s_1]}
=(O_{\tmpl})_{\#}\mu^{\gate}_{L,\bar c[s\leftarrow s_1]}$ and
$\Delta_{\gate}=0$. The conclusion concerns this site-level criterion only, and
\Cref{lem:site-necessary} does not upgrade it to $\Phi_{\mathrm{NI}}$.
\end{proof}

The last hypothesis is the load-bearing one and Route I evidence does not supply
it. It says the model's choice of the slice inputs is itself secret-independent,
which is noninterference of $M$. It holds for a deterministic library call whose
arguments are fixed by the code, and it fails for the mail skill of
\Cref{sec:motivation}, where $M$ builds $\mathit{body}$ from the context.
\Cref{thm:empty-unsound} and \Cref{prop:code-resolved} therefore bracket the
space from opposite sides without meeting in the middle.

\subsection{Certified Separation}

\begin{proof}[Proof of \Cref{thm:cert-separation}]
Clause (a) uses no semantic premise about executions. Acceptance is atomic with
respect to the record append, so the freshness tests, the condition checks, and
the store advance share one linearization point.

Conditions~1 to 3 of \Cref{app:verify} fix the objects. Live mode with
$\mathrm{res}=\bot$ lets the same operation reserve $q_t$ and append
\texttt{DECIDE(accept)}, which is the theorem's protocol clause. The certificate
names the current execution-state version $\sigma_t$ and the immutable
\texttt{PROPOSE} record of $q_t$ with no decision record, so $\widehat G^{-}_t$
is exactly the active graph named, and the current epoch fixes
$\widehat{\Sev},\widehat{\Impl},\widehat{\terminalbad}$ and the schema. Binding to
the proposal, and not to a decision field, breaks the circularity.

Conditions~4 to 7 rebuild $\widehat B_t$, $\widehat P_t$, and $\widehat U_t$ from
that fixed material, recompute every binding and capacity, and check that
removing $\cut_t$ separates $\widehat U_t$ from
$\widehat B_t\cup\widehat P_t$ at the declared cost. Conditions~8 to 10 tie the
cut to enforcement, giving $\mathrm{dom}(\manifest_t)=\cut_t$ with one entry per
edge, no $\bot$ binding, canonical hook metadata, and declared joint
admissibility. Minimality is the one property the verifier takes on declaration
rather than recomputation, which is why \Cref{thm:cert-separation} states the
declared cost and not an optimality claim. For clause (b), acceptance reruns
all nine obligations on each effect atom and its capability witness. Their
coverage, validator, declassifier, and hook premises then establish authority
for the concrete effects of the checked call.
\end{proof}

\subsection{The Cut: Closed Form}
\label{app:proof-cut}

The network is the flow subgraph of $G^{-}_t$ with a super-source joined to
$B_t\cup P_t$ and a super-sink joined from $U_t$, both at infinite capacity, so
max-flow/min-cut applies directly \citep{fordfulkerson} and the connectors never
appear in a minimum cut. Each original edge appears once and is charged once. The
emission rules make the resulting optimum explicit.

\begin{proposition}[Closed form of the certified cut]
\label{prop:cut-closed-form}
Fix a step with $U_t=\{T(q_t)\}$ and write $\kappa(C)=\sum_{e\in C}\kappa(e)$.
Let $A$ be the argument edges of $a_t$ and $e_{\mathrm T}$ its terminal edge
$\mathrm{op}(q_t)\to T(q_t)$. Let $R$ be the set of nodes reachable from
$B_t\cup P_t$ in the committed flow subgraph, \emph{including
$B_t\cup P_t$ itself}, and let $A^{\star}=\{e\in A:\src(e)\in R\}$. If
$A^{\star}=\emptyset$ then $\cut_t=\emptyset$. Otherwise, provided at least one of
$e_{\mathrm T}$ and $A^{\star}$ is bindable throughout,
\begin{equation}
  \kappa(\cut_t)=\min\bigl\{\,\kappa(\{e_{\mathrm T}\}),\ \kappa(A^{\star})
  \,\bigr\},
  \label{eq:cut-closed}
\end{equation}
and $\cut_t$ is the lexicographically least minimum-cost separating set under the
canonical edge order. Placing $e_{\mathrm T}$ before every argument edge in that
order makes $\cut_t$ equal to $\{e_{\mathrm T}\}$ on a tie. One reachability pass
over the committed flow subgraph therefore determines $\cut_t$ in time linear in
the size of $G^{-}_t$, given $B_t$ and the bindings.
\end{proposition}

\begin{proof}
By \Cref{eq:capacity} the only finite-capacity edges are the proposed flow
instances of $a_t$ with a nonempty binding, and by (E1) and (E4) of
\Cref{app:coverage} those lie in $A\cup\{e_{\mathrm T}\}$; rules (E2), (E3), and
(E5) fire at finalization and emit committed edges only. The node
$\mathrm{op}(q_t)$ is fresh, its only out-edge before dispatch is
$e_{\mathrm T}$, and $T(q_t)$ has no other in-edge. Every directed path from
$B_t\cup P_t$ to $T(q_t)$ therefore ends
$\src(e)\to\mathrm{op}(q_t)\to T(q_t)$ for some $e\in A$ with $\src(e)\in R$,
that is for some $e\in A^{\star}$; the reflexive case, where $\src(e)$ is itself
protected or untrusted, is the path of length two. Conversely each
$e\in A^{\star}$ carries such a path. If $A^{\star}=\emptyset$ no path exists and
$\emptyset$ separates. Otherwise let $C$ be any separating set of finite-capacity
edges. For each $e\in A^{\star}$, cutting the path through $e$ requires
$e\in C$ or $e_{\mathrm T}\in C$, so either $e_{\mathrm T}\in C$ or
$A^{\star}\subseteq C$, and by strict positivity of $\kappa$ we get
$\kappa(C)\geq\min\{\kappa(\{e_{\mathrm T}\}),\kappa(A^{\star})\}$ with equality
only for $C=\{e_{\mathrm T}\}$ or $C=A^{\star}$. Both of these separate, so they
are the only minimizers and \Cref{eq:cut-closed} follows. If some edge of
$A^{\star}$ is unbindable then $\kappa(A^{\star})=\infty$ and the minimum is
attained at $\{e_{\mathrm T}\}$ whenever $e_{\mathrm T}$ is bindable; if
$e_{\mathrm T}$ is unbindable too, the minimum is infinite, the sentinel appears
in every separating set, and R5 refuses.
\end{proof}

Only the current call's edges are severable, so the optimization reduces to a
two-way choice between blocking the call and paying for its live argument edges,
and \Cref{lem:cuttable} shows the first branch always separates when the adapter
offers \texttt{block\_call}. Optimality therefore does not imply usefulness: the
cut earns its place only where the argument branch is both cheaper and available,
which the occurrence warrant makes rare for model-chosen calls.
The preceding proof also gives the final-action criterion in
\Cref{sec:guarantee}: for any $S\subseteq A\cup\{e_{\mathrm T}\}$, every
unsafe path is hit iff $e_{\mathrm T}\in S$ or $A^\star\subseteq S$.
A restoration executes with $S=\emptyset$; when $A^\star\ne\emptyset$, at
least one path survives regardless of C's verdict. We keep the max-flow formulation because it
survives emission rules that add intermediate proposed edges, for which the
two-candidate argument no longer applies.

\subsection{Coverage and Simulation on the Transparent Fragment}
\label{app:proof-transparent}

\begin{proof}[Proof of \Cref{prop:transparent}]
Write $w_i$ for the call that realizes the link $x_i\!\to\!x_{i+1}$ for $i<m-1$,
$q_t$ for the call that realizes the final link $x_{m-1}\!\to\!\zeta_t$, and
$\nu$ for the host's provenance map, and set $\alpha_{\mathrm{D}}=\nu$ on the
values of $\gamma$ and $\alpha_{\mathrm{T}}(\zeta_t)=T(q_t)$. The argument is an
induction over the host trace using only (E1) to (E5) of \Cref{app:coverage},
which the mediation layer executes and the artifact does not.

\emph{Base and the direct case.} $x_0$ has a provenance identifier by
hypothesis. If $m=1$ it enters $q_t$ at a declared position, so E1 and E4 emit
the final segment directly and the value-link induction is empty. If $m>1$,
transparency puts $x_0$ at a declared argument position $j$ of $w_0$, coverage
gives $\mathrm{schema}(\mathrm{op}(w_0),j)\neq\bot$, and E1 emits
$\alpha_{\mathrm{D}}(x_0)\!\to\!\mathrm{op}(w_0)$ before $w_0$ leaves the host.

\emph{Value links.} Fix $i\geq0$ and suppose every earlier segment is emitted.
Transparency puts $x_{i+1}$ either at a declared return position of $w_i$ or in a
resource $w_i$ declares that it writes. In the first case (E2) mints its
identifier and emits $\mathrm{op}(w_i)\!\to\!\alpha_{\mathrm{D}}(x_{i+1})$ at the
outcome hook, giving the two-edge segment
$\pi_i=\alpha_{\mathrm{D}}(x_i)\!\to\!\mathrm{op}(w_i)\!\to\!
\alpha_{\mathrm{D}}(x_{i+1})$; this is the shortest form available, since the
construction has no \texttt{Data}-to-\texttt{Data} edge. In the second case the link is realized by the write $w_i$ together with the read
$w_i'$ that follows it. Rule (E5) mints the resource version $R_{v+1}$ and commits
$\mathrm{op}(w_i)\!\to\!R_{v+1}$ with the union of that call's committed input
labels, and $w_i'$ contributes $R_{v+1}\!\to\!\mathrm{op}(w_i')$ by (E1) and
$\mathrm{op}(w_i')\!\to\!\alpha_{\mathrm{D}}(x_{i+1})$ by (E2), so $\pi_i$ has
four edges and passes through $R_{v+1}$. The proposition's finalization hypothesis names both, $w_i$ as the call that
produced the resource version and $w_i'$ as the call that produced the value. Either way, identifier
preservation and (E3) carry the identifier through any pass-through tool, and the
transport hypothesis puts the value at the declared argument position of the next
call, so the in-edge (E1) emits starts where the out-edge of $w_i$ ended. A weaker
notion of identity fails here: byte equality would merge unrelated values, and
object identity would not survive serialization at the tool boundary.

\emph{The final effect link.} Rules (E1) to (E3) cannot supply this step, since
they emit only \texttt{Data}--\texttt{Operation}--\texttt{Data} segments and
never reach a \texttt{Terminal} node. By transparency $x_{m-1}$ enters $q_t$ at a
declared position, so (E1) emits
$\alpha_{\mathrm{D}}(x_{m-1})\!\to\!\mathrm{op}(q_t)$, and by (E4) the proposal
extension creates $T(q_t)$ fresh and emits $\mathrm{op}(q_t)\!\to\!T(q_t)$, both
at \systemhook{before\_tool\_call}. Their composition
$\pi_{m-1}=\alpha_{\mathrm{D}}(x_{m-1})\!\to\!\mathrm{op}(q_t)\!\to\!T(q_t)$
reaches the effect and not only the last value. Without (E4) the proposition
would establish value-to-value provenance and say nothing about confinement.

\emph{Concatenation, timeliness, and entailment.} Consecutive segments share the
endpoint $\alpha_{\mathrm{D}}(x_{i+1})$, so they concatenate into one flow path
from $\alpha_{\mathrm{D}}(x_0)$ to $T(q_t)$. The calls $w_0,\dots,w_{m-2}$
and every reading call $w_i'$ appended \texttt{OUTCOME} before $q_t$ was proposed,
so their proposal, output, and resource-version edges are committed and appear in
$G^{-}_t$, whether those calls were certified or permissively released. The two
edges of $\pi_{m-1}$ are \texttt{proposed} instances of $q_t$ present before
dispatch. Taking a direct value link requires its declared input and output, which
fire (E1) and (E2). Taking a resource-mediated link requires the declared write
and the declared read, which fire (E1) and (E5) for $w_i$ and (E1) and (E2) for
$w_i'$, that is all four edges of $\pi_i$. The no-secondary-call condition
excludes any other realization of a link, at value links and at the final link
alike. Rule (E4) fires before dispatch and the schema assigns traversal of the
terminal edge to crossing that mediated effect boundary. Taking a link therefore
entails its segment, and $\gamma$ is simulated.
\end{proof}

\begin{remark}[Reach of the transparent fragment]
\label{rem:transparent-scope}
Hidden input and hidden effect escape the fragment; internal computation does
not. The fragment is narrow because identifiers do not survive the model, so
\Cref{sec:experiments} reports its coverage over tools, calls, and whole chains.
Coverage is stated apart from H1 to keep the pair from being circular: coverage is
structural and audited offline, while emission and simulation are runtime claims
inside H1. A covered link whose emitter never fires is therefore a violated
premise, not a silently uncovered one. Source and terminal labeling, H3, is
measured on its own.
\end{remark}

\subsection{Freshness, Authentication, and Anti-Replay}
\label{app:proof-freshness}

\begin{proposition}[One-shot authorization and one-shot consumption]
\label{prop:recertification}
Assume A1 to A8. Every accepted certificate advances the verifier store by
exactly one. A certificate naming a stale committed state or an old sequence
number is rejected, and one accepted once cannot be accepted again. A token
issued once is admitted by the executor at most once. And after a crash, the
absence of \texttt{DISPATCH} implies no external send.
\end{proposition}

\begin{proof}
By (A5) acceptance is one linearization point: the reservation test, the
freshness tests, the update
$h^{\mathrm{auth}}\gets\mathsf{h}(h^{\mathrm{auth}}\|\Cert_t)$,
\texttt{DECIDE(accept)}, and token issue are inseparable, so no record or crash
leaves an accepted-but-undecided reservation, and two submissions of the same
$n_t$ cannot both succeed. Condition~2 of \Cref{app:verify} rejects stale
authorization. A certificate with $n_t\neq n^{\ast}\!+\!1$ is rejected, so the
sequence is a strict successor relation, not an equality test against a
current digest, which would reject every new certificate while admitting the most
recent one. A certificate whose $\sigma_t$ is not the current execution-state version
is rejected as well; by (A6) the raw log tip is not tested, since the step's own
records move it. Replay of an accepted certificate is caught by condition~2, which
finds a decision already recorded for that attempt, and by the sequence test; a
\texttt{BLOCKED} or \texttt{CANCEL} attempt does not advance $\sigma_t$, so the
state test alone would not reject it. By (A4) the certificate cannot be
re-presented as a fresh attempt without a different $q_t$ and hence a different
token preimage.

Verifier acceptance alone would not bound execution, since an issued token could
be replayed while the first call is still running. The executor therefore admits
$\mathsf{tok}_t$ by a compare-and-swap on its own store keyed by $q_t$, which by
(A4) to (A6) admits it at most once and survives restart. This is also why $q_t$
is never reused: a retry under the same identifier would collide with a consumed
entry.

Binding ties the token to the manifest. By (A1) no party without $k'$ produces a
valid $\mathsf{tok}_t$; by (A3) the preimage
$q_t\|n_t\|\sigma_t\|\varepsilon_t\|\mathsf{h}(L^{\mathrm{p}}_t)\|
\mathsf{h}(\manifest_t)$ parses uniquely, so a producer cannot shift bytes
between fields; and by (A2) a producer cannot show one manifest to the verifier
and a different manifest with the same digest to the executor. Without (A2) or
(A3) the executor could install a manifest the verifier never checked while the
token still verifies.

The crash claim needs (A7) and (A8). A finalized policy block has a durable
\texttt{BLOCKED} terminal and no send, and a failed preparation sends nothing and
finalizes \texttt{CANCEL}. Otherwise, were \texttt{DISPATCH} appended only in
memory, a request could leave the host while the crash erased the record, and
recovery would then classify a completed irreversible effect as never dispatched
and permit an automatic retry. Write-ahead durability with no reordering rules
that out, so absence of \texttt{DISPATCH} implies no send and permits
\texttt{CANCEL}.
\end{proof}

\begin{proposition}[The executed manifest is certified against the executed
state]
\label{prop:view-freshness}
Under A1 to A6, A8, and A9, no other mediated call is reserved, dispatched, or
committed between the moment $\textsc{AcceptIfFresh}$ reserves $q_t$ and the
moment $\textsc{Finalize}$ releases it. At dispatch the immutable proposal,
resource snapshots or leases, execution-state version, policy epoch, and installed
manifest are exactly those the certificate names, and the dispatched call is
exactly $\mathrm{Apply}_{\varepsilon_t}(\manifest_t,L^{\mathrm p}_t)$.
\end{proposition}

\begin{proof}
By \Cref{thm:cert-separation} the token commits to $\sigma_t$, $\varepsilon_t$,
$\mathsf{h}(L^{\mathrm{p}}_t)$, and $\mathsf{h}(\manifest_t)$. A5 and A6 make the
global reservation linearizable and rollback-protected, so no other mediated call
can advance $\sigma_t$ before finalization, and by A9 no unmediated event can
change the committed evidence without advancing $\sigma_t$ either. By A8,
$\textsc{Prepare}$ spends the token only after rechecking $q_t$, $\sigma_t$,
$\varepsilon_t$, and the two digests against the immutable \texttt{PROPOSE}
record, and its \texttt{READY} branch dispatches nothing but
$\mathrm{Apply}_{\varepsilon_t}(\manifest_t,L^{\mathrm p}_t)$ while holding every
resource lease through the certified access. Its
$\textsc{LeaseIfCurrent}(\varepsilon_t)$ atomically checks the epoch and acquires
a read lease; a changed epoch finalizes \texttt{CANCEL} without dispatching, and
after success a writer must wait. By collision resistance of $\mathsf{h}$ the
installed manifest is the one whose digest the token carries, and preparation
returns \texttt{READY} only once all entries are installed. Clauses three and four
of H5 are therefore consequences of the protocol rather than assumptions, which is
why \Cref{thm:confinement} may cite this proposition. Dropping A9 leaves the
statement true of the store and false of the world.
\end{proof}

\begin{remark}[Why the reservation is required]
\label{rem:single-flight}
The epoch freezes configuration, not execution state, so without a global
reservation this interleaving is admitted. Call $q_1$ reads a shared file and
sends it to the network, and at certification the file is benign, so the cut is
empty or cheap. Before $q_1$ dispatches, $q_2$ writes a protected value into that
file and commits. No configuration changed, so $\varepsilon$ did not move and the
old token still verifies against the epoch, the certified cut never saw the new
path, and $q_1$ sends the secret. Two properties close it: the token carries
$\sigma_t$, which $q_2$'s \texttt{OUTCOME} advances, so the executor's recheck
turns the interleaving into a refusal, and the reservation prevents it from
arising at all, covering permissive releases too. Separately,
\Cref{prop:recertification} bounds how often a \emph{token} is spent, not how
often an external effect happens, which is why \texttt{DISPATCH} is committed
before the call is invoked and why recovery retries only in the provably
undispatched case.
\end{remark}

\subsection{Grounded Confinement and the Episode Invariant}
\label{app:proof-episode}

\begin{proof}[Proof of \Cref{thm:confinement}]
Let $\gamma=x_0\!\to\!\cdots\!\to\!x_{m-1}\!\to\!\zeta_t$ be a schema-covered
real causal chain from a registered protected source to the pending unsafe
effect, and suppose the verifier accepted its certificate but the chain
completes. For $i<m-1$ let $w_i$ realize $x_i\to x_{i+1}$ with simulated segment
$\pi_i$, let $\pi_{m-1}$ be the final segment for $q_t$, and let
$\alpha_{\mathrm D},\alpha_{\mathrm T}$ be the simulation maps with
$\alpha_{\mathrm T}(\zeta_t)=T(q_t)$. By H3a, registration gives
$\alpha_{\mathrm{D}}(x_0)\in P_t$. By H3c, the unsafe-effect rule gives
$T(q_t)\in U_t$. By H1 every link has a nonempty flow segment, the segments
concatenate, and taking a real link entails traversing every edge of its segment.
Write $\Pi=\pi_0\cdots\pi_{m-1}$, a flow path from $P_t$ to $T(q_t)$ in
$G^{-}_t$. Concatenation is what makes $\Pi$ a witness of reachability, since
per-link emission alone would leave segments that never meet at shared endpoints.
Because $G^{-}_t$ is folded in phase~1 from the committed state and the proposal
only, this witness exists before the call runs and uses no outcome of step $t$.

By \Cref{thm:cert-separation}, $T(q_t)$ is unreachable from $P_t$ after removing
$\cut_t$, so some edge $e$ of $\Pi$ lies in $\cut_t$. That edge lies on at least
one segment; choose a corresponding real link, which by entailment cannot be
taken unless $e$ is traversed. The same theorem gives
$\mathrm{dom}(\manifest_t)=\cut_t$ with one entry per cut edge, so $e$ has an
entry, and that entry is the canonical minimum-cost member of
$\Sev(e)\cap\Impl(e)$ on a \texttt{proposed} instance with the registry's
argument and hook. By H4 the executor installs the entry and its severance
contract holds under the complete manifest, and by the definition of $\Sev$ in
\Cref{sec:cut} that contract speaks about real execution, not about the graph. By
the theorem's freshness condition the manifest, state, and epoch at dispatch are
the certified ones, which \Cref{prop:view-freshness} implements under its protocol
assumptions. The installed entry is therefore the entry checked, the link cannot
be taken, and $\gamma$ does not complete, contradicting the supposition. The scope
is exactly coverage: a chain with an uncovered link has no declared segment,
witnesses no reachability, and is untouched.
\end{proof}

\begin{corollary}[Grounded integrity]
\label{cor:integrity}
Replace H3a by H3b and let $\gamma$ start at a live value the deployment marks
untrusted. Under the same conditions $\gamma$ cannot complete at this call.
\end{corollary}

\begin{proof}
The proof above uses $\alpha_{\mathrm{D}}(x_0)\in P_t$ only to place the path's
source in the cut's source set. H3b places it in $B_t$ instead, and
\Cref{thm:cert-separation} separates $U_t$ from $B_t\cup P_t$, so the rest of the
argument is unchanged. The integrity statement is weaker in practice because the
occurrence warrant is unavailable for model-chosen calls, so the cut that
discharges it is almost always $\{e_{\mathrm T}\}$.
\end{proof}

\begin{corollary}[Monitor invariant over an episode]
\label{cor:episode}
Let the certified steps of an episode be exactly $t_1<\cdots<t_N$, with H1, H3a,
H3c, H4, and H5 at each of them and A1 to A9 throughout. Then
\Cref{thm:confinement} holds at every $t_i$. Suppose in addition that the episode
runs under \method{}-\textsc{c} and that H2 holds for the class of effects at
which a covered chain could complete, so the deployment registers every one of
them as a mediated unsafe boundary. Then no covered chain from a registered
protected source completes at any certified step.
\end{corollary}

\begin{proof}
The first part applies \Cref{thm:confinement} at each $t_i$. For the second, H2
routes every effect of the registered class through the mediator, so such an
effect executes only under an accepted decision. A decision is accepted only
inside an accepting $\textsc{AcceptIfFresh}$, and by
\Cref{prop:recertification} a certificate is accepted at most once and its
authorization admitted at most once. Under \method{}-\textsc{c} there is no
uncertified executing branch: R1, R4, R5, and R6 refuse, and the remaining branch
certifies. A
certified step with an empty cut installs nothing. But H3c makes
$U_t\neq\emptyset$ wherever a covered chain could complete, so
\Cref{prop:cut-closed-form} applies and an empty cut means
$A^{\star}=\emptyset$. A covered chain would put $\alpha_{\mathrm D}(x_0)\in
P_t\subseteq R$, carry its prefix over committed edges so that
$\alpha_{\mathrm D}(x_{m-1})\in R$, and place its last value at a declared
argument position of $q_t$, which forces that argument edge into $A^{\star}$. So
no covered chain completes at an empty-cut step either. Persistence of previously installed entries carries
nothing, since after the committed state changes an old cut may be bypassed by a
new edge. Execution is prevented by token discipline instead, which makes this a
reference-monitor invariant and not a further confinement argument. H2 is a
completeness assumption on the deployment's effect registry, and without it the
statement says only that nothing leaks through boundaries the deployment labelled
correctly.

The corollary is stated for \method{}-\textsc{c} because
\method{}-\textsc{p} has executing branches without a grounded path
certificate. Under the semantically sound labels of \Cref{app:scope}, an
outward or irreversible effect cannot satisfy R2's three conjuncts. Outside R2,
an unrepaired dispatch, including R3 and a restored P block, requires all nine
effect obligations to pass; otherwise the only executing route is a
registry-declared repair of a repairable verdict. This proves the
payload-independent claim in \Cref{sec:cut}: an out-of-authority outward or
irreversible effect can execute only through such a repair, given sound request
compilation, schema coverage and semantics, validators and declassifiers,
complete mediation of the effects, and application of every required hook.
Because repairs are not rechecked by R6,
the argument does not certify their final effects. R2 and restoration likewise
do not establish the episode confinement invariant. A leaking R2 call would
require an unsound label or a missing sink/persistence declaration;
reversibility alone is not its admission test. The per-call theorem applies
to \method{}-\textsc{p} only when its actual installed action satisfies all the
theorem's premises.
\end{proof}

\section{Experimental Details}
\label{app:experiments}

\subsection{Frozen Hyperparameters}
\label{app:hyperparameters}

Every value below is fixed before the primary run and is not tuned per benchmark,
model, or seed.

\begin{table}[t]
\caption{Frozen evaluation and policy configuration.}
\centering
\small
\begin{tabularx}{\linewidth}{@{}p{0.33\linewidth}X@{}}
\toprule
Field & Value \\
\midrule
Target models &
\texttt{gpt-5.6-luna}; \texttt{DeepSeek-V4-Flash};
local \texttt{Qwen-3.8-27B} served by vLLM \\
Decoding & Temperature $0$; maximum completion length $1024$ tokens \\
Ablation seed & \texttt{20260828}; single deterministic paired run \\
Concurrency & 8 workers; throughput only, no effect on the statistical
population \\
Defense lanes &
\texttt{none}, \texttt{pace}, and the external baselines of
\Cref{tab:baselines-app} \\
Layer switches & Path confinement, capability and effect verification, and
execution-boundary adaptation, toggled independently for the ablation arms
of \Cref{tab:arms-app} \\
Cut solver & $\src$--$\snk$ max-flow/min-cut over the flow subgraph of $G^-_t$ \\
Effect obligations & The nine checks of \Cref{sec:cut}, evaluated in a fixed
order \\
Judge & Frozen prompt and aggregation; evidence required; no manual
correction \\
Technical-failure ceiling & $5\%$ per reported cell \\
\bottomrule
\end{tabularx}
\end{table}

\subsection{Benchmarks, Adapters, and Denominators}
\label{app:benchmark-details}

\begin{table}[h]
\caption{Benchmarks, native evaluators, and the paired-subset size used in the
ablation.  Commit and data revisions are recorded in the release manifest of
\Cref{app:reproducibility}.}
\label{tab:benchmarks-app}
\centering
\small
\begin{tabularx}{\linewidth}{@{}p{0.20\linewidth}Xr@{}}
\toprule
Benchmark & Scaffold and native evaluator & Paired cases \\
\midrule
AgentDojo & Official task, security, and utility functions & 198 \\
AgentDyn & Official three-domain runner and utility scorer & 60 \\
WASP & Official end-to-end evaluator; plain-text and URL groups & 24 \\
ASB & Official agents and metric scripts; seven run variants & 280 \\
PASB & Official personalized-agent workflow; IPI and four memory subsets & 58 \\
InjecAgent & Official tools and cases; Base and Enhanced splits & 212 \\
MCPTox & Official templates with the native refusal metric & 135 \\
MSB & Official stage scorer with PUA and NRP & 200 \\
\midrule
Total & & 1167 \\
\bottomrule
\end{tabularx}
\end{table}

Each benchmark keeps its own success and utility functions.  AgentDojo and
AgentDyn use executable security and utility; InjecAgent reports direct-harm and
data-stealing rates; ASB its five attack groups with utility; WASP the plain-text
and URL-injection groups with end-to-end user utility; PASB the IPI group and
four memory subsets; MCPTox per-template attack success with a refusal rate; and
MSB stage scores with PUA and NRP. Capabilities are compiled deterministically
from the authenticated request, intersecting request-derived effect families,
trusted argument constraints, organizational policy, and tool schemas. The
compiler reads no injected artifact, attack label, or judge output. Policies and
schemas were developed from public interfaces, task formats, and development
tests, then frozen before the final runs and completion of the factorial arms.

For every adapter the harness stores an episode record containing the user task,
the injected artifact, the target-model revision, the tool schema, every call to
the defense hooks, tool results, selected actions, certificate-verifier decisions, native
evaluator output, judge evidence, latency, token count, and terminal status.  For a defended run, an episode lacking a required defense-hook call is invalid
rather than a sample. The additional adaptive search uses the same AgentDojo
version and configuration. An authority audit of all 629 (user, injection) pairs
checks ground-truth target calls without dispatch: 375 (59.6\%) fail AUTH,
comprising 60/144 Banking, 39/105 Slack, 93/140 Travel, and 183/240 Workspace
pairs. With seed 20260922, ten out-of-authority pairs are selected from each of
Slack, Travel, and Workspace; Banking is excluded. The undefended agent and
\method{}-\textsc{p} are attacked separately on
\texttt{Qwen/Qwen-3.8-27B-FP8} at temperature 0. Starting from the static
injection, a mutator using that checkpoint at temperature 1.0 sees six candidates
sampled from the eight best and four most recent, the agent's tool calls, guard
decisions and first failed obligations where applicable, and the security
verdict; it proposes four new payloads per round. Each pair stops at first
success or 50 episodes. The protocol, sample, and code were fixed before an
isolated five-pair pilot. Twenty of the 21 undefended successes, selected with a
fixed seed, were reviewed manually.

Each reported rate carries four quantities: the numerator, the valid denominator,
the invalid count, and the technical-failure count.  Three conventions matter for
reading the tables.  InjecAgent DS-S2 is a conditional metric, so a cell with no
eligible sample is reported as \texttt{N/A (0 eligible)} and never as a measured
zero.  MCPTox refusal and ASB refusal use the valid denominator of their own
arm, which differs across arms whenever the invalid count differs.  For MSB, group NRP is $N_g=U_g(1-A_g)$ on rates in $[0,1]$.
Full-table companion metrics pool their eligible samples, using
$M_{\rm pool}=\sum_g n_{g,M}M_g/\sum_g n_{g,M}$; undefined PUA/NRP groups
contribute no samples. Thus a weighted mean of group NRPs need not equal the
product of separately displayed PUA and ASR means. NRP is a composite score,
not a counted joint safe-and-successful outcome. The official type-macro output
instead averages ASR across valid attack types and PUA across types where it is
defined, then computes
$\mathrm{PUA}_{\rm macro}(1-\mathrm{ASR}_{\rm macro})$. Full type coverage has 12 ASR and 8 PUA types;
\texttt{false\_error}, \texttt{simulated\_user}, and their injection
combinations lack native PUA. The ablation's four-stage macro is a separate
summary and is not interchangeable with the full-table pooling.

\subsection{Ablation Design}
\label{app:ablation-design}

\begin{table}[h]
\caption{Ablation arms.  P is path confinement, C is capability and effect
verification, B is execution-boundary adaptation.  All eight arms run on the same
paired subset, so component contrasts are paired; zero-ASR floors can prevent identifying
a security benefit.}
\label{tab:arms-app}
\centering
\small
\begin{tabularx}{\linewidth}{@{}llllX@{}}
\toprule
Arm & P & C & B & Role \\
\midrule
A0 & & & & Wrapper-matched no-defense corner \\
A1 & \checkmark & & & Path confinement alone \\
A2 & & \checkmark & & Capability and effect verification alone \\
A3 & & & \checkmark & Execution-boundary adaptation alone \\
A4 & \checkmark & \checkmark & & Path confinement with verification \\
A5 & \checkmark & & \checkmark & Path confinement with boundary adaptation \\
A6 & & \checkmark & \checkmark & Verification with boundary adaptation \\
A7 & \checkmark & \checkmark & \checkmark & Complete \method{}, paired anchor \\
\bottomrule
\end{tabularx}
\end{table}

Every arm keeps the runner, the transport, the case selector, and the evaluation
chain of the complete system and changes only the layer switches, so a reduced
arm is wrapper-matched to complete \method{} rather than compared against a
historical no-defense run.  The complete-method column is re-aggregated offline
from the run records of the full system on the frozen case identifiers, which
issues no new model call and leaves the main tables untouched.

Two data-handling notes belong here rather than in the main text.  Ten MSB cases
had failed for a dependency reason in an earlier run and were re-executed under
the same configuration and merged back by case identifier, which brings both
affected arms to 200 scored cases with no technical failure.  On ASB the paired
selection is made by matching a stable trajectory identity, namely the agent
name, the attack tool, and the initial user message; an earlier aggregation that
took the first forty rows of each run file has been withdrawn, and the audited
selection resolves all seven run variants one to one with no missing or ambiguous
match. The completed factorial contains all 64 benchmark--arm cells.
C toggles capability installation, effect expansion, verification, and the
combiner together, so its contrasts concern that bundle. Metric-specific
eligibility is applied after matching logical case IDs. ASB and
InjecAgent-Enhanced have zero ASR already at A0 and identify no security
improvement.

\subsection{Statistical Analysis}
\label{app:statistics}

The main and per-group tables report point estimates on native denominators;
they do not display confidence intervals. For the adaptive case-success result
we report $x/n$ and a two-sided 95\% Wilson interval.  For correlated actions we resample
complete episodes, and where several attacks share a task we resample the task
cluster.  Paired comparisons use identical case identifiers, target revisions,
and decoding settings, with an exact McNemar test and Holm correction specified for
inferential comparisons inside each family; no such significance claim is made
from the displayed point estimates.  The ablation is one deterministic paired run under a single
seed, so we report differences on that population and make no multi-seed
stability claim. The displayed arm rates remain descriptive comparisons;
we make no claim of a statistically significant component effect.  We do not pool attack success across benchmarks, since the
native success functions are not commensurable, and we do not describe any
difference as statistically significant without the corresponding test.

Technical failure includes endpoint timeout after the frozen retry policy,
malformed official evaluator state, missing hook invocation, unrecoverable tool
sandbox failure, or an absent native denominator.  It is reported as its own
rate.  A cell above the $5\%$ ceiling is removed from the primary table and
retained, with failure causes, in \Cref{app:additional-results}.

The frozen judge returns an unsafe or safe label, the decisive trajectory span,
the affected resource, and the native-metric mapping.  The manifest records official evaluator outputs and judge evidence; the current
tables do not provide a numerical judge-agreement or parse-failure analysis.  Judge labels are never edited by hand.

\subsection{Baseline Port Fidelity}
\label{app:port-fidelity}

\begin{table}[h]
\caption{External baselines by family, with the benchmarks each is reported on
and the fidelity of the port.}
\label{tab:baselines-app}
\centering
\small
\begin{tabularx}{\linewidth}{@{}p{0.20\linewidth}p{0.30\linewidth}X@{}}
\toprule
Family & Methods & Benchmarks and fidelity \\
\midrule
No defense & Baseline, ReAct & All eight; native \\
Trusted flow and IFC &
CaMeL \citep{camel}, FIDES \citep{fides}, DRIFT \citep{drift} &
AgentDojo, AgentDyn, WASP, PASB, MSB; native on AgentDojo, validated port
elsewhere \\
Tool-use safeguards &
Progent \citep{progent}, Tool Filter \citep{agentdojo},
Tool Allowlist \citep{mcptox}, MCIP Guardian \citep{mcip},
ToolShield \citep{toolshield} &
AgentDojo, ASB, AgentDyn, MCPTox, MSB, PASB; native or validated port \\
Re-execution & MELON \citep{melon} & AgentDojo; native \\
Detection &
DataSentinel \citep{datasentinel},
PI-Detector \citep{agentdojo}, PIGuard \citep{piguard},
PromptGuard-2 \citep{promptguard2},
Metadata Sanitization \citep{mcptox} &
AgentDojo, InjecAgent, ASB, AgentDyn, MCPTox, PASB, MSB; pre-filter plus
end-to-end \\
Prompt-level &
Delimiters, Sandwich, Spotlighting, Instructional Prevention,
Direct and PoT Paraphrase, PoT Shuffle, Repeat User Prompt &
WASP, ASB, InjecAgent, AgentDyn, PASB; native \\
Diagnostic reference & DTA \citep{taintdroid,libdft} & AgentDojo; faithful
instrumentation \\
\bottomrule
\end{tabularx}
\end{table}

DTA propagates labels through instrumented hashes and persistent sessions and
serves as an instrumented diagnostic reference, excluded from main-table
selection. DataSentinel is evaluated end to end as a pre-filter. All reported
comparisons are local executions with pinned method and benchmark revisions;
manifests retain model identities, configurations, intervention positions,
evaluators, and output provenance. Ports adapt APIs, tool schemas, and endpoints
while preserving the method's intervention. Within a benchmark and target model,
methods share the case universe, attack collection, and native evaluation chain.
On AgentDojo, the undefended agent, \method{}, CaMeL, DTA, Progent, MELON,
DataSentinel, Spotlighting, and PI-Detector each completed 12/12 jobs across
three attacks and four suites. Unsupported method--task combinations remain
outside their subgroup averages; execution completion alone is not a valid
denominator for every native metric.

\subsection{Reproducibility Manifest and Replay}
\label{app:reproducibility}

The artifact contains a machine-readable manifest with the following fields:
\begin{verbatim}
paper_protocol_revision
repository_commit
container_digest
benchmark_name, benchmark_commit, data_revision
agent_scaffold_revision, adapter_revision
target_provider, target_model, target_revision
attack_family, generator_model, generator_prompt_hash
judge_model, judge_prompt_hash, judge_temperature
seed, defense_lane, policy_config_hash
episode_id, graph_version
native_metric_numerators, native_metric_denominators
technical_failure_code, latency_ms, token_counts
\end{verbatim}
The manifest is the authoritative record of every revision referenced in
\Cref{app:benchmark-details}. For each
episode, replay reconstructs the initial graph, consumes the stored
sequence of proposals and tool events, and checks every graph version,
cut, binding, policy action, certificate-verifier
decision, and executor receipt.  A replay is successful only if all
deterministic fields match.  Model generations are stored as inputs to this
replay rather than re-sampled.  A separate live replay mode reissues model
calls and is reported as a non-deterministic replication.

\subsection{Cost, Release Checklist, and Artifact Safety}

The harness records target-agent tokens, judge tokens, wall-clock latency,
flow-solver time, peak graph size, certificate-check time, and tool sandbox CPU
and memory.  A numerical mediation-overhead summary is not reported in this manuscript.  The solver is deterministic for a fixed graph, so model-call variance is
reported separately.  Hardware, vLLM revision, quantization status, tensor
parallelism, container digest, and the API price snapshot are frozen in the
reproducibility manifest.

\begin{enumerate}[leftmargin=1.5em,itemsep=0.15em]
\item Freeze all eight benchmark and scaffold revisions.
\item Publish the complete policy configuration, the effect schema, and the
capability compiler.
\item Publish the agent prompt and the frozen judge prompt.
\item Release per-episode native evaluator inputs and outputs, subject to
benchmark licenses and removal of operational secrets.
\item Release verifier tests for stale state, stale prefix, non-separating cut,
misstated cut cost, unbindable edge, and each of the nine effect obligations.
\item Release the notebook that regenerates every table and interval from the
result manifest, including the frozen case selectors of the paired subset.
\end{enumerate}

The public artifact contains no live credentials, production endpoints, or
working exfiltration targets.  Harmful payloads are represented by inert local
sinks and deterministic fixtures that preserve the path structure the evaluation
needs.

\section{Extra Experiments}
\label{app:additional-results}

\subsection{Per-Group Benchmark Results}
\label{app:full-tables}

\Cref{tab:main} report, for each benchmark and model, the worst
eligible attack group with its companion metric. Each block lists the undefended
agent, \method{}, and the three comparators with lowest three-model mean
worst-group ASR; ties use the companion metric. DTA is excluded as a diagnostic
reference. For InjecAgent the summary uses Direct Harm and DS-S1 for every
method, excluding conditional DS-S2 and the derived data-stealing aggregate.
Across the eight tables there are 84 elementary ASR columns; removing the three
DS-S2 columns and two with zero undefended ASR leaves 79. \method{} is strictly
lowest in 62, tied in 14, and higher in three: Qwen WASP plain text (4.76\%
versus FIDES 2.38\%, with 92.86\% versus 0.0\% utility), Qwen MSB-Call
(2.50\% versus Tool Filter 0.0\%, with 60.14\% versus 47.34\% PUA), and
DeepSeek AgentDyn GitHub (0.06\% versus CaMeL 0.0\%, with 63.57\% versus
3.75\% utility; PI-Detector and DRIFT also attain 0.0\% ASR). The tables
retain every group and ported baseline on the same runs as the main summaries;
metric-specific eligibility follows \Cref{app:benchmark-details}.

\begin{table*}[htbp]
\centering
\caption{Results on WASP under the official end-to-end evaluator. Attack success is
reported for the plain-text and URL-injection formats; utility is the native user-task
accuracy. Lower is better for ASR, higher for utility; best ASR per block in bold.}
\resizebox{\textwidth}{!}
{
\begin{tabular}{l|ccc|ccc|ccc}
\toprule
 \multirow{5}{*}{\textbf{Method}} & \multicolumn{3}{c|}{\textbf{Qwen-3.8-27B}} & \multicolumn{3}{c|}{\textbf{DeepSeek-V4-Flash}} & \multicolumn{3}{c}{\textbf{gpt-5.6-luna}} \\
\midrule
&\textbf{Plain-text} & \textbf{URL Injection} & \textbf{Utility} & \textbf{Plain-text} & \textbf{URL Injection}& \textbf{Utility} & \textbf{Plain-text} & \textbf{URL Injection}&\textbf{Utility}\\
 &\textbf{ASR} $\downarrow$ & \textbf{ASR} $\downarrow$ & \textbf{Accuracy} $\uparrow$ & \textbf{ASR} $\downarrow$ & \textbf{ASR} $\downarrow$ & \textbf{Accuracy} $\uparrow$ & \textbf{ASR} $\downarrow$ & \textbf{ASR} $\downarrow$ & \textbf{Accuracy} $\uparrow$\\
\midrule 
Baseline      &19.05         & 14.29      & 78.57         & 9.52  &7.14   & 79.76 & 4.76  &4.76   & 86.90  \\
Prompt Filter &9.52          & 14.29      & 82.14         & 7.14  &4.76   & 75.00 & 11.90 &9.52   & 83.33  \\
Simple Static &20.37         & 15.32      & 72.64         & 16.67 &14.29  & 64.29 & 19.04 &14.29  & 84.52 \\
FIDES         &2.38          & 40.48      & 0.0           & 19.05 &30.95  & 32.14 & 4.76  &23.81  & 59.52 \\
\midrule
PACE(Ours)    &4.76 &\textbf{7.14}&\textbf{92.86} & \textbf{0.0} &\textbf{0.0}&\textbf{91.67} & \textbf{0.0} &\textbf{0.0}&\textbf{95.24}\\
\bottomrule 
\end{tabular}
}
\label{tab:wasp}
\end{table*}

\begin{table*}[htbp]
\centering
\caption{Results on AgentDojo under the official security and utility functions,
by attack type. Lower is better for ASR, higher for utility; best ASR per block in bold.}
\resizebox{\textwidth}{!}
{
\begin{tabular}{l|cccc|cccc|cccc}
\toprule
 \multirow{5}{*}{\textbf{Method}} & \multicolumn{4}{c|}{\textbf{Qwen-3.8-27B}} & \multicolumn{4}{c|}{\textbf{DeepSeek-V4-Flash}} & \multicolumn{4}{c|}{\textbf{gpt-5.6-luna}} \\
\midrule
&\textbf{Direct} & \textbf{Import Instructions}&\textbf{Tool Knowledge} & \textbf{Utility} & \textbf{Direct} & \textbf{Import Instructions}&\textbf{Tool Knowledge}& \textbf{Utility} & \textbf{Direct} & \textbf{Import Instructions}&\textbf{Tool Knowledge}&\textbf{Utility} \\
 &\textbf{ASR} $\downarrow$ & \textbf{ASR} $\downarrow$& \textbf{ASR} $\downarrow$ & \textbf{Accuracy} $\uparrow$ & \textbf{ASR} $\downarrow$ & \textbf{ASR} $\downarrow$& \textbf{ASR} $\downarrow$ & \textbf{Accuracy} $\uparrow$ & \textbf{ASR} $\downarrow$ & \textbf{ASR} $\downarrow$ & \textbf{ASR} $\downarrow$& \textbf{Accuracy} $\uparrow$ \\
\midrule 
Baseline     &48.00        &30.21       & 30.52     &62.59  & 97.62 & 99.68 & 99.84 &75.68 & 99.05&99.84 &100   &73.77 \\
CaMeL        &0.32         & 0.32       & 0.32      &35.14  & 20.51 & 24.48 & 28.14 &60.04 & 31.00&36.09 &34.02 &57.02 \\
DTA          &0.48         & 0.95       & 0.95      &29.84  & 24.17 & 28.46 & 33.33 &54.00 & 42.13&48.01 &55.01 &56.01 \\
Spotlighting &98.40        & 86.30      & 80.00     &62.47  & 100   & 92.05 & 90.14 &75.42 & 100  &100   &100   &70.01 \\
PI Detector  &99.50        & 96.00      & 90.30     &44.40  & 100   & 100   & 92.05 &67.99 & 100  &100   &99.68 &71.01\\
MELON        &0.0          & 12.70      & 1.59      &49.76  & 16.06 & 50.08 & 20.03 &57.98 & 61.05&43.08 &54.05 &52.52\\
Progent      &0.64         & 2.54       & 3.02      &57.76  & 35.14 & 42.13 & 61.53 &64.02 & 48.01&53.10 &69.00 &62.94\\
DataSentinel &3.02         & 35.29      & 30.37     &73.77  & 62.00 & 78.06 & 74.09 &85.00 & 84.10&80.92 &75.99 &82.95\\
\midrule
PACE(Ours)  &\textbf{0.0}&\textbf{0.0}&\textbf{0.0}&59.67&\textbf{5.09}&\textbf{2.70}&\textbf{3.02}&75.25&\textbf{8.11}&\textbf{6.83}&\textbf{7.79}&72.05 \\
\bottomrule 

\end{tabular}
}
\label{tab:agentdojo}
\end{table*}

\begin{table*}[htbp]
\centering
\caption{Results on the InjecAgent-Enhanced split. Direct Harm and the two
data-stealing stages DS-S1 and DS-S2 are native attack-success rates, as is the aggregate
data-stealing column. DS-S2 is conditional and excluded, together with the
derived aggregate, from the main-table summary and 79-column comparison.
Lower is better; best per block in bold.}
\resizebox{\textwidth}{!}
{
\begin{tabular}{l|cccc|cccc|cccc}
\toprule
 \multirow{5}{*}{\textbf{Method}} & \multicolumn{4}{c|}{\textbf{Qwen-3.8-27B}} & \multicolumn{4}{c|}{\textbf{DeepSeek-V4-Flash}} & \multicolumn{4}{c|}{\textbf{gpt-5.6-luna}} \\
\midrule
&\textbf{Direct Harm} & \textbf{DS-S1}&\textbf{DS-S2} & \textbf{Data-stealing} &\textbf{Direct Harm} & \textbf{DS-S1}&\textbf{DS-S2} & \textbf{Data-stealing} &\textbf{Direct Harm} & \textbf{DS-S1}&\textbf{DS-S2} & \textbf{Data-stealing}\\
 &\textbf{ASR(\%)} $\downarrow$ & \textbf{ASR(\%)} $\downarrow$& \textbf{ASR(\%)} $\downarrow$ & \textbf{ASR(\%)} $\downarrow$ & \textbf{ASR(\%)} $\downarrow$ & \textbf{ASR(\%)} $\downarrow$& \textbf{ASR(\%)} $\downarrow$ & \textbf{ASR(\%)} $\downarrow$ & \textbf{ASR(\%)} $\downarrow$ & \textbf{ASR(\%)} $\downarrow$ & \textbf{ASR(\%)} $\downarrow$& \textbf{ASR(\%)} $\downarrow$ \\
\midrule 
Baseline                 &12.35 & 23.71 & 79.07 &18.75  & 19.61 &28.13 &52.29 &14.71 &12.55 &18.57&50.50&9.38 \\
ReAct                    &24.31 & 27.94 & 92.11 &25.74  & 27.84 &37.68 &79.02 &29.77 &16.86 &23.35&66.14&15.44 \\
PromptGuard              &13.14 & 24.45 & 80.45 &19.67  & 21.57 &31.80 &64.74 &20.59 &10.00 &17.28&56.38&9.74 \\
Prompt-Sandwich          &0.59  &  0.18 & 0.15  & 0.12  & 4.71  &14.89 &50.62 &7.54  &2.16  &17.10&73.12&12.50 \\
Instructional Prevention &1.37  &  3.13 & 58.82 & 1.83  & 9.61  &19.67 &67.29 &13.24 &3.73  &15.26&56.63&8.64 \\
PI Detector              &3.13  & 11.03 & 20.00 & 2.21  & 13.92 &27.21 &56.08 &15.26 &5.29  &23.71&57.36&13.60 \\
DataSentinel             &13.14 & 24.45 & 80.45 &19.67  & 22.16 &25.18 &74.45 &18.75 &17.06 &12.13&65.15&7.90\\
PromptGuard-2            &12.75 & 23.89 & 80.77 &19.30  & 20.59 &26.29 &87.41 &22.98 &11.57 &19.49&82.06&15.99\\
\midrule
PACE(Ours)& \textbf{0.0}&\textbf{0.0}&\textbf{0.0}&\textbf{0.0}&\textbf{0.98}&\textbf{3.12}&\textbf{0.0} &\textbf{0.0} &\textbf{0.59} &\textbf{2.39} &\textbf{30.77}&\textbf{0.74} \\
\bottomrule 

\end{tabular}
}
\label{tab:injectagent}
\end{table*}

\begin{table*}[htbp]
\centering
\caption{Results on ASB across its five attack settings, with the native utility score.
Lower is better for ASR, higher for utility; best ASR per block in bold.}
\resizebox{\textwidth}{!}
{
\begin{tabular}{l|cccccc|cccccc|cccccc|}
\toprule
 \multirow{5}{*}{\textbf{Method}} & \multicolumn{6}{c|}{\textbf{Qwen-3.8-27B}}& \multicolumn{6}{c|}{\textbf{DeepSeek-V4-Flash}} & \multicolumn{6}{c|}{\textbf{gpt-5.6-luna}} \\
\midrule
&\textbf{DPI} & \textbf{OPI}&\textbf{MP} & \textbf{Mixed} & \textbf{PoT} & \textbf{Utility}&\textbf{DPI} & \textbf{OPI}&\textbf{MP} & \textbf{Mixed} & \textbf{PoT} & \textbf{Utility} & \textbf{DPI} & \textbf{OPI}&\textbf{MP} & \textbf{Mixed} & \textbf{PoT} & \textbf{Utility}\\
 &\textbf{ASR(\%)} $\downarrow$ & \textbf{ASR(\%)} $\downarrow$& \textbf{ASR(\%)} $\downarrow$ & \textbf{ASR(\%)} $\downarrow$ & \textbf{ASR(\%)} $\downarrow$ & \textbf{Accuracy(\%)} $\uparrow$& \textbf{ASR(\%)} $\downarrow$ & \textbf{ASR(\%)} $\downarrow$& \textbf{ASR(\%)} $\downarrow$ & \textbf{ASR(\%)} $\downarrow$ & \textbf{ASR(\%)} $\downarrow$ & \textbf{Accuracy(\%)} $\uparrow$& \textbf{ASR(\%)} $\downarrow$ & \textbf{ASR(\%)} $\downarrow$& \textbf{ASR(\%)} $\downarrow$ & \textbf{ASR(\%)} $\downarrow$ & \textbf{ASR(\%)} $\downarrow$ & \textbf{Accuracy(\%)} $\uparrow$\\
\midrule 
Baseline                  & 84.92        &66.00      & 13.25      & 88.50       & 56.25        & 30.29 & 74.08 &58.25 &36.50 & 70.75 & 26.00 & 64.00    &81.08      & 31.50      & 27.00       & 72.75        & 73.50 & 80.25  \\
Delimiters                & 88.08        &73.75      & 13.50      & 86.75       & 23.75        & 24.61 & 88.42 &85.50 &51.75 &93.00 &38.00   &54.39     &81.63      & 68.25      & 48.25       & 88.00        & 76.25 & 66.57 \\
Sandwich                  & 85.33        &61.75      & 11.75      & 88.50       & 68.75        & 32.50 & 81.83 &76.25 &46.00 &59.75 &25.75 & 69.07      &71.21      & 61.50      & 52.00       & 63.25        & 57.25 & 87.86 \\
Instructional Prevention  & 71.92        &60.75      & 17.00      & 82.00       & 52.75        & 33.93 & 58.75 &61.75 &40.75 &78.50 &20.75  &72.50      &54.57      & 40.75      & 43.25       & 83.75        & 46.00 & 87.25 \\
Direct Paraphrase         & 60.08        &71.50      & 13.00      & 49.50       & 55.50        & 62.00 & 68.08 &79.00 &49.50 &51.50 &43.25 &81.64       &51.07      & 60.75      & 40.50       & 70.75        & 81.75 & 81.57 \\
PoT Paraphrase            & 84.50        &70.00      & 13.75      & 89.00       & 59.00       & 30.82 & 78.58 &87.75 &53.75 &66.25 &56.75 &62.96              &61.50      & 82.00      & 50.25       & 43.75        &92.00 & 82.43 \\
PoT Shuffle               & 84.92        &71.75      & 13.00      & 88.50       & 99.00       & 28.57 & 65.42 &43.25 &20.75 &45.75 &23.50 &57.43              &57.43      & 51.00      & 34.50       & 59.25       & 43.00 & 67.04\\
\midrule
PACE(Ours)                &\textbf{0.42}&\textbf{0.00}&\textbf{0.00}&\textbf{0.00}&\textbf{0.00}& 29.57 & \textbf{7.08}  &\textbf{4.75} &\textbf{7.25} &\textbf{10.5}&\textbf{3.25} &64.64 &\textbf{5.07}&\textbf{3.75}&\textbf{5.50}&\textbf{8.75}&\textbf{6.75}&81.63 \\
\bottomrule 
\end{tabular}
}
\label{tab:asb}
\end{table*}

\begin{table*}[htbp]
\centering
\caption{Results on AgentDyn across its three task domains, with the native utility
score. Lower is better for ASR, higher for utility; best ASR per block in bold.}
\resizebox{\textwidth}{!}
{
\begin{tabular}{l|cccc|cccc|cccc|}
\toprule
 \multirow{5}{*}{\textbf{Method}} & \multicolumn{4}{c|}{\textbf{Qwen-3.8-27B}}& \multicolumn{4}{c|}{\textbf{DeepSeek-V4-Flash}} & \multicolumn{4}{c|}{\textbf{gpt-5.6-luna}} \\
\midrule
&\textbf{Shopping} & \textbf{Github}&\textbf{DailyLife} & \textbf{Utility} & \textbf{Shopping} & \textbf{Github}&\textbf{DailyLife} & \textbf{Utility}&\textbf{Shopping} & \textbf{Github}&\textbf{DailyLife} & \textbf{Utility}\\
 &\textbf{ASR(\%)} $\downarrow$ & \textbf{ASR(\%)} $\downarrow$& \textbf{ASR(\%)} $\downarrow$ & \textbf{Utility(\%)} $\uparrow$ & \textbf{ASR(\%)} $\downarrow$ & \textbf{ASR(\%)} $\downarrow$& \textbf{ASR(\%)} $\downarrow$ & \textbf{Utility(\%)} $\uparrow$ & \textbf{ASR(\%)} $\downarrow$ & \textbf{ASR(\%)} $\downarrow$& \textbf{ASR(\%)} $\downarrow$ & \textbf{Utility(\%)} $\uparrow$ \\
\midrule 
Baseline                  & 11.11 & 17.22 & 32.50 & 63.39 & 12.78 & 17.22 & 19.50 & 65.00 & 7.78 &10.56 & 19.50 & 74.82   \\
Repeat User Prompt        & 6.67  & 15.56 & 34.50 & 61.25 & 22.78 & 10.00 & 31.50 & 61.43 & 12.28&10.00 & 31.50 & 70.54 \\
Spotlighting              & 3.33  & 13.89 & 5.50  & 65.00 & 13.33 & 7.22  & 3.00  & 58.21 & 17.22&15.00 & 27.50 & 64.46  \\
Tool Filter               & 0.0   & 0.0   & 0.05  & 7.86  & 1.67  & 2.78  & 2.00  & 8.04  & 2.78 &3.89  & 6.50  & 9.11   \\
PI-Detector               & 0.0   & 0.0   & 0.05  & 1.25  & 0.0   & 0.0   & 3.50  & 0.08  & 7.22 &1.67  & 2.50  & 2.14   \\
PI-Guard                  & 0.0   & 0.0   & 5.50  & 6.79  & 2.78  & 6.67  & 13.50 & 6.60  & 3.89 &8.33  & 1.00  & 3.93  \\
PromptGuard-2             & 6.11  & 15.56 & 33.00 & 58.42 & 9.44  & 17.22 & 23.50 & 60.36 & 3.33 &12.78 & 25.50 & 68.57   \\
CaMeL                     & 0.0   & 0.0   & 0.0   & 0.0   & 1.11  & 0.0   & 5.50  & 3.75  & 0.0  & 0.0  &  0.0  & 0.0 \\
Progent                   & 1.67  & 1.67  & 11.50 & 26.96 & 5.56  & 1.67  & 8.50  & 9.82  & 2.22 &1.67  & 5.50  & 11.25 \\
DRIFT                     & 0.0   & 0.56  & 0.50  & 29.82 & 2.78  & 0.0   & 1.50  & 13.39 & 0.0  &3.33  & 4.50  & 14.64 \\
\midrule
PACE(Ours)& \textbf{0.0} & \textbf{0.0}&\textbf{0.0}&60.54&  \textbf{0.0} & 0.06&\textbf{0.05}&63.57 &\textbf{0.0} & \textbf{0.0}&\textbf{0.0}&73.57 \\
\bottomrule 
\end{tabular}
}
\label{tab:AgentDyn}
\end{table*}

\begin{table*}[htbp]
\centering
\caption{Results on PASB for indirect prompt injection and the four memory subsets, with
the native IPI utility. Lower is better for ASR, higher for utility; best ASR per block in bold.}
\resizebox{\textwidth}{!}
{
\begin{tabular}{l|cccccc|cccccc|cccccc|}
\toprule
 \multirow{5}{*}{\textbf{Method}} & \multicolumn{6}{c|}{\textbf{Qwen-3.8-27B}}& \multicolumn{6}{c|}{\textbf{DeepSeek-V4-Flash}} & \multicolumn{6}{c|}{\textbf{gpt-5.6-luna}} \\
\midrule
&\textbf{IPI}&\textbf{Short Read} & \textbf{Short Modify}&\textbf{Long Read} & \textbf{Long Modify}& \textbf{Utility}&\textbf{IPI}&\textbf{Short Read} & \textbf{Short Modify}&\textbf{Long Read} & \textbf{Long Modify}& \textbf{Utility}&\textbf{IPI} & \textbf{Short Read} & \textbf{Short Modify}&\textbf{Long Read} & \textbf{Long Modify}& \textbf{Utility}\\
 &\textbf{ASR(\%)} $\downarrow$ & \textbf{ASR(\%)} $\downarrow$& \textbf{ASR(\%)} $\downarrow$ & \textbf{ASR(\%)} $\downarrow$& \textbf{ASR(\%)} $\downarrow$ & \textbf{Accuracy(\%)} $\uparrow$ & \textbf{ASR(\%)} $\downarrow$& \textbf{ASR(\%)} $\downarrow$ & \textbf{ASR(\%)} $\downarrow$ & \textbf{ASR(\%)} $\downarrow$ & \textbf{ASR(\%)} $\downarrow$& \textbf{Accuracy(\%)} $\uparrow$& \textbf{ASR(\%)} $\downarrow$& \textbf{ASR(\%)} $\downarrow$& \textbf{ASR(\%)} $\downarrow$ & \textbf{ASR(\%)} $\downarrow$& \textbf{ASR(\%)} $\downarrow$ & \textbf{Accuracy(\%)} $\uparrow$\\
\midrule 
Baseline                  & 0.0 & 85.00 & 62.50 & 62.50 & 75.00 & 98.47 & 0.0 &67.50 &72.50 & 60.00 & 92.50 &99.24 & 2.29 & 55.00&42.50 &50.00 &97.5 &100.0  \\
Delimiters                & 0.0 & 25.00 & 0.0   & 45.00 & 20.00 & 99.23 & 0.0 &30.00 &55.00 & 42.50 & 32.50 &98.47 & 1.53 & 22.50 & 35.00 &30.00 & 35.00 & 99.24 \\
Sandwich                  & 0.0 & 85.00 & 70.00 & 65.00 & 72.50 & 99.24 & 0.0 &72.50 &85.00 & 57.50 & 87.50 &98.47 & 3.82 & 40.00 & 25.00 &42.50 &82.50 &100.0 \\
Instructional Prevention  & 0.0 & 17.50 & 0.0   & 17.50 & 0.0   & 97.67 & 0.0 & 20.00&17.50 &27.50 &17.50 &96.95 &6.87 & 17.50 &7.50 & 27.50 &25.00 &98.47 \\
PI-Detector               & 0.0 & 85.00 & 50.00 & 70.00 & 67.50 & 99.24 & 0.0 & 77.50&65.00 &60.00 &82.50 &100.0 &2.29 & 67.50 & 35.00 & 40.00 & 70.00 & 98.47 \\
PI-Guard                  & 0.0 & 87.50 & 65.00 & 70.00 & 67.50 & 100.0 & 0.0 &82.50 &77.50 &55.00 &87.50 &99.27 &3.05 & 77.50 & 45.00 & 50.00 & 80.00 & 96.95 \\
PromptGuard-2             & 0.0 & 85.00 & 45.00 & 62.50 & 57.50 & 100.0 & 0.0 &72.50 &55.00 &50.00 &72.50 &97.71 & 83.97 & 70.00 & 40.00 & 32.50 & 62.50 & 99.24 \\
Progent                   & 0.0 & 80.00 & 5.00  & 65.00 & 12.50 & 100.0 & 0.0 &77.50 &12.50 &67.50 &22.50 &100.0 &12.98 &62.50 & 20.00 &62.50 &35.00 &98.47  \\
FIDES                     & 0.0 & 0.0   & 27.50 & 0.0   & 40.00 & 99.24 & 0.0 &5.00 & 15.00& 2.50&35.00 &98.47 &14.50 & 22.50 &32.50&12.50&65.00&99.27 \\
ToolShield                & 0.0 & 47.50 & 42.50 & 57.50 & 55.00 & 98.44 & 0.0 &50.00 &35.00 &55.00 &60.00 &95.42 &9.16 & 62.50&42.50 &70.00 &77.50&96.18  \\
CaMeL                     & 0.0 & 71.75 & 13.00 & 88.50 & 99.00 & 28.57 & 0.0 &7.50 &15.00 &37.50 &52.50 &98.47 & 5.34 &35.00 &62.50 &32.50 &67.50&99.24  \\
\midrule
PACE(Ours)& \textbf{0.0} & \textbf{0.0}&\textbf{0.0} & \textbf{0.0}& \textbf{0.0}&97.69&\textbf{0.0}  &\textbf{0.0} & \textbf{0.0}&\textbf{2.50} & \textbf{5.0}&99.24&\textbf{0.0}&\textbf{0.0} & \textbf{2.50}&\textbf{0.0} & \textbf{2.50}&100.0 \\
\bottomrule 

\end{tabular}
}
\label{tab:PASB}
\end{table*}

\begin{table*}[htbp]
\centering
\caption{Results on MCPTox for the three poisoning templates, with the native refusal
rate. Lower is better for both; refusal is reported so that a low attack-success rate cannot
be read without its cost. Best ASR per block in bold.}
\resizebox{\textwidth}{!}
{
\begin{tabular}{l|cccc|cccc|cccc|}
\toprule
 \multirow{5}{*}{\textbf{Method}} & \multicolumn{4}{c|}{\textbf{Qwen-3.8-27B}}& \multicolumn{4}{c|}{\textbf{DeepSeek-V4-Flash}} & \multicolumn{4}{c|}{\textbf{gpt-5.6-luna}} \\
\midrule
&\textbf{Template-1} & \textbf{Template-2}&\textbf{Template-3}&\textbf{Refusal} & \textbf{Template-1} & \textbf{Template-2}&\textbf{Template-3}&\textbf{Refusal} & \textbf{Template-1} & \textbf{Template-2}&\textbf{Template-3}&\textbf{Refusal} \\
 &\textbf{ASR(\%)} $\downarrow$ & \textbf{ASR(\%)} $\downarrow$& \textbf{ASR(\%)} $\downarrow$& \textbf{Refusal Rate(\%)} $\downarrow$ & \textbf{ASR(\%)} $\downarrow$ & \textbf{ASR(\%)} $\downarrow$ & \textbf{ASR(\%)} $\downarrow$& \textbf{Refusal Rate(\%)} $\downarrow$  & \textbf{ASR(\%)} $\downarrow$& \textbf{ASR(\%)} $\downarrow$ & \textbf{ASR(\%)} $\downarrow$& \textbf{Refusal Rate(\%)} $\downarrow$  \\
\midrule 
Baseline                   & 22.39 & 8.84 & 32.52 & 0.16 & 56.54 & 52.61&56.06 &0.42 &23.04 &26.10 &86.93 &27.96  \\
Metadata Sanitization      & 5.94  & 1.38 & 17.41 & 0.15 & 13.09 & 9.91 &28.60 &0.37 &5.76  &5.43  &91.48 &26.67  \\
PI-Detector                & 20.53 & 7.37 & 25.22 & 0.40 & 85.34 & 4.59 &50.19 &1.63 &48.17 &3.97  &95.45 &30.14  \\
PI-Guard                   & 11.59 & 3.09 & 11.07 & 0.91 & 50.79 & 27.56&30.68 &3.15 &34.03 &9.19  &63.64 &46.32  \\
PromptGuard-2              & 21.43 & 7.68 & 27.08 & 0.63 & 63.35 & 55.95&53.98 &1.68 &60.73 &37.16 &97.72 &31.02  \\
Tool Allowlist             & 0.72  & 0.0  & 0.18  & 2.46 & 15.18 & 13.57& 2.84 &7.66 &6.81  &5.85  &16.29 &55.78  \\
MCIP Guardian              & 0.51  & 0.40 & 0.19  & 54.23& 9.95  & 5.01 & 1.89 &76.72&7.33  &3.55  &10.04 &81.22  \\
Progent                    & 1.52  & 1.98 & 8.80  & 45.88& 20.42 & 15.24&16.48 &66.85&19.90 &11.27 &46.40 &78.43 \\
\midrule
PACE(Ours)& \textbf{0.0} & \textbf{0.0}&\textbf{0.0} & 1.97& \textbf{2.62} & \textbf{1.46}&\textbf{0.76} & 0.58& \textbf{0.52} & \textbf{0.0}&\textbf{2.65} & 28.55 \\
\bottomrule 

\end{tabular}
}
\label{tab:MCPTox}
\end{table*}

\begin{table*}[htbp]
\centering
\caption{Results on MSB by attack stage, with the native PUA and NRP scores. Attack
success is lower better; PUA and NRP are higher better. Cells marked \texttt{--} were not
produced by the official scorer for that configuration.}
\resizebox{\textwidth}{!}
{
\begin{tabular}{l|cccccc|cccccc|cccccc|}
\toprule
 \multirow{5}{*}{\textbf{Method}} & \multicolumn{6}{c|}{\textbf{Qwen-3.8-27B}}& \multicolumn{6}{c|}{\textbf{DeepSeek-V4-Flash}} & \multicolumn{6}{c|}{\textbf{gpt-5.6-luna}} \\
\midrule
&\textbf{Plan} & \textbf{Call}&\textbf{Response} & \textbf{Multi-Stage} & \textbf{PUA} & \textbf{NRP}&\textbf{Plan} & \textbf{Call}&\textbf{Response} & \textbf{Multi-Stage} & \textbf{PUA} & \textbf{NRP} & \textbf{Plan} & \textbf{Call}&\textbf{Response} & \textbf{Multi-Stage} & \textbf{PUA} & \textbf{NRP}\\
 &\textbf{ASR(\%)} $\downarrow$ & \textbf{ASR(\%)} $\downarrow$& \textbf{ASR(\%)} $\downarrow$ & \textbf{ASR(\%)} $\downarrow$ & \textbf{Accuracy(\%)} $\uparrow$ & \textbf{Accuracy(\%)} $\uparrow$ & \textbf{ASR(\%)} $\downarrow$ & \textbf{ASR(\%)} $\downarrow$& \textbf{ASR(\%)} $\downarrow$ & \textbf{ASR(\%)} $\downarrow$ & \textbf{Accuracy(\%)} $\uparrow$ & \textbf{Accuracy(\%)} $\uparrow$ & \textbf{ASR(\%)} $\downarrow$ & \textbf{ASR(\%)} $\downarrow$& \textbf{ASR(\%)} $\downarrow$ & \textbf{ASR(\%)} $\downarrow$ & \textbf{Accuracy(\%)} $\uparrow$ & \textbf{Accuracy(\%)} $\uparrow$\\
\midrule 
Baseline        & 5.67 & 38.75 & 26.35 & 18.70 & 46.34 & 36.29 &61.33&100.0 &29.68 &52.00 &91.24 &49.29 &67.00 &98.75 &9.03 &50.20 &89.08 &54.51 \\
MCIP Guardian   & 13.33& 50.00 & 14.03 & 8.40  & 59.15 & 55.70 &8.00 &77.50 &26.13 &37.30 &96.55 &51.32 &29.33 &76.25 &8.55 &36.90 &92.34 &56.87 \\
PI-Detector     & 16.67& 38.75 & 7.02  & 5.20  & 54.36 & 49.83 &55.33&53.75 &18.55 &48.40 &94.81 &51.18 &65.00 &70.00 &7.10 &47.10 &90.85 &56.22 \\
PI-Guard        & 14.00& 37.50 & 13.87 & 9.00  & 56.42 & 53.26 &32.00&57.50 &36.94 &36.90 &91.98 &49.95 &40.67 &66.25 &10.00 &38.50&89.97 &54.25 \\
PromptGuard-2   & 14.67& 38.75 & 25.00 & 17.60 & 47.85 & 40.87 &61.67&93.75 &68.06 &55.20 &93.27 &45.57 &78.00 &87.50 &16.29 &53.60&91.66 &50.58 \\
Tool Filter     & 15.00& 0.0   & 30.81 & 14.20 & 47.34 & 43.48 &9.67 &42.50 &59.35 &51.70 &93.65 &47.46 &22.33 &52.50 &13.71 &50.60&91.04 &49.75 \\
Progent         & 11.00& 38.75 & 5.65  & 4.30  & 50.98 & 47.37 &47.67&97.50 &26.45 &30.60 &97.08 &46.67 &58.33 &96.25 &7.90  &34.40&90.51 &47.64 \\
FIDES           & 1.33 & 52.50 & 10.16 & 10.60 & 45.88 & 40.40 &29.00&100.0 &23.71 &49.80 &90.66 &52.62 &34.67 &100.0 &8.87  &48.20&87.73 &56.45 \\
ToolShield      & 2.71 & 38.75 & 7.02  & 5.20  & 53.94 & 45.89 &4.33 &100.0 &20.16 &56.00 &96.24 &53.58 &9.00 &100.0 &6.77  &65.20 & -- & -- \\
\midrule
PACE(Ours)& \textbf{0.0}&\textbf{2.50}& \textbf{0.0}& \textbf{0.0}&\textbf{60.14}&\textbf{60.01}&\textbf{0.0}&\textbf{3.75}& \textbf{0.16}& \textbf{1.40}&\textbf{98.83}&\textbf{59.84}&\textbf{0.33} &\textbf{3.75}&\textbf{0.0}&\textbf{0.7}&\textbf{94.87} &\textbf{60.61} \\
\bottomrule 

\end{tabular}
}
\label{tab:MSB}
\end{table*}

\subsection{Ablation Notes}

\tightparagraph{The ASB utility row.}
\Cref{tab:ablation} reports 0/280 ASB utility for six of the eight arms, with
only P{+}C and complete \method{} at 84/280.  The exact-pair audit confirms that
these zeros come from the official \texttt{Original Task Successful} column and
are not an aggregator artifact that turned missing values into zeros.  The audit
settles the question of sample selection and does not by itself explain why the
official scorer returns zero for those six arms.  We therefore read that row
as an implementation-sensitive diagnostic and interpret it only together with
trajectory completion, the refusal column, and the native scorer semantics.  The
same audit withdraws the earlier paired values of 65/280 utility and 106/280
refusal, which corresponded to the first forty rows of each run file rather than
to the forty trajectories the frozen selector names.

\tightparagraph{Coverage and certificate health.}
The quantities in \Cref{tab:cert-health-app} are instrumentation fields,
not numerical estimates of semantic coverage. The A7 records contain 32867
PACE decisions across eight benchmarks, 26272 with full witnesses across six;
MCPTox and MSB retain compact records. All nine obligations are evaluated on
25972 schema-covered decisions, yielding 233748 predicate evaluations,
including failures. There are 22685 nonempty candidate cuts, 5911 terminal
blocks, 365 rewrite-bearing decision records, and 22191 named adapter
activations. These count different objects and cannot be summed; a rewrite field
is not proof that a rewrite executed. AgentDyn has 6647 guarded-call records,
including failures and recovery attempts: every candidate path certificate
passes, C certifies 6190, and 5842 have nonempty cuts containing 6080
edges. No nonterminal rewrite executes, so restored utility there comes from
the combiner overriding P's terminal blocks. Existing tests include positive and
negative cases for each obligation, forged witnesses and capabilities, and 20
mutations each for token forgery, schema drift, stale certificates, parameter
smuggling, and omitted adapters. These check interface behavior, not semantic
completeness or distributed durability.

\tightparagraph{Reporting conventions.}
The displayed benchmark cells are rounded point estimates; native numerators
and denominators belong to the run records.  Provider refusal, tool
failure, policy intervention, and native task failure are separate terminal
states, so a safe model refusal is never attributed to \method{} and a broken
tool episode never counts as a defense success.  Cells excluded from a primary
table under the technical-failure ceiling of \Cref{app:statistics} are listed
with endpoint timeout, evaluator error, missing hook, sandbox failure, and
malformed certificate as mutually exclusive causes, with retried and non-retried
denominators both recorded. In the reduced-scale adaptive experiment,
undefended successes are 10/10 Slack, 7/10 Travel, and 4/10 Workspace,
totalling 21/30 (70.0\%; two-sided 95\% Wilson interval [52.1, 83.3]\%).
Twelve first succeed with the static seed and nine with search-generated
payloads at iterations 3--39, using 632 episodes in total.
\method{}-\textsc{p} yields 0/30 (interval [0, 11.4]\%) after 1500 episodes.
Its boundary trace records 157 exact out-of-authority target calls in 126
episodes across 14 pairs, all blocked before execution: 149 fail AUTH and
BUDGET, and eight additionally fail BIND and CTRL. None uses R2.
Two controlled checks are excluded from ASR: direct boundary dispatch blocks
all 37 frozen target-call patterns, and submitting a target call in fresh
contexts carrying payloads from 1003 selected target-free episodes yields
the same refusal. These do not increase the 30-case attack sample. The sample
covers only out-of-authority goals in three suites, uses the target checkpoint
as mutator, and permits 50 episodes per case; stronger mutators and attacks on
parsing or unrepresented transformations remain untested.

\begin{table}[h]
\caption{Certificate-health quantities recorded per episode.}
\label{tab:cert-health-app}
\centering
\small
\begin{tabularx}{\linewidth}{@{}p{0.37\linewidth}p{0.23\linewidth}X@{}}
\toprule
Quantity & Statistic & Interpretation \\
\midrule
\texttt{forward\_simulation\_}\newline\texttt{observed=true} & Rate + CI &
Last event required no graph expansion; not proof of H1 \\
\texttt{op.unknown} expansions & Histogram/quantiles &
Observed abstraction incompleteness \\
Unbindable minimum cuts & Rate + CI &
Graph separation lacking an executor action \\
Stale-state rejection & Rate + CI &
Freshness check exercised before dispatch \\
Effect-obligation failures & Histogram by obligation &
Which of the nine checks fails, and how often \\
Bound action executed & Rate + CI &
Observable executor agreement relevant to H4 \\
Verifier reconstruction mismatch & Rate + cause &
Certificate construction or serialization defect \\
\bottomrule
\end{tabularx}
\end{table}

\end{document}